\documentclass[lettersize,journal]{IEEEtran}
\ifCLASSINFOpdf
\else
\fi

\usepackage[T1]{fontenc}
\DeclareUnicodeCharacter{00A0}{~}
\DeclareFontFamily{T1}{calligra}{}
\DeclareFontShape{T1}{calligra}{m}{n}{<->s*[1.44]callig15}{}
\DeclareMathAlphabet\mathcalligra   {T1}{calligra} {m} {n}
\DeclareMathAlphabet\mathzapf       {T1}{pzc} {mb} {it}
\DeclareMathAlphabet\mathchorus     {T1}{qzc} {m} {n}
\DeclareMathAlphabet\mathrsfso      {U}{rsfso}{m}{n}
\usepackage{graphicx}
\usepackage{amsmath}
\usepackage{subfigure}
\usepackage{amsfonts,amsmath}
\usepackage{amssymb, nccmath}
\usepackage{amssymb, nccmath}
\usepackage{dsfont}
\usepackage{booktabs}
\usepackage{multirow}
\usepackage[utf8]{inputenc}

\usepackage{dsfont}
\usepackage{booktabs}
\usepackage{graphicx}
\usepackage{grffile}
\usepackage{subfig}
\usepackage{grffile}
\usepackage{array} 
\usepackage{mathtools}
\usepackage{amsmath}
\usepackage{amsthm}
\usepackage{fontenc}
\usepackage{tikz}
\usepackage{lineno}
\usepackage{cite}
\usepackage{subfig}
\usepackage[linesnumbered, ruled]{algorithm2e}
\makeatletter
\newcommand{\removelatexerror}{\let\@latex@error\@gobble}
\makeatother

\newtheorem{theorem}{Theorem}
\newtheorem{corollary}{Corollary}{}
\newtheorem{definition}{Definition}[]
\makeatletter
\def\ps@IEEEtitlepagestyle{%
	\def\@oddfoot{\mycopyrightnotice}%
	\def\@oddhead{\hbox{}\@IEEEheaderstyle\leftmark\hfil\thepage}\relax
	\def\@evenhead{\@IEEEheaderstyle\thepage\hfil\leftmark\hbox{}}\relax
	\def\@evenfoot{}%
}
\def\mycopyrightnotice{%
	\begin{minipage}{\textwidth}
		\centering \scriptsize
		This article has been accepted in IEEE Transactions on Systems, Man, and Cybernetics: Systems Journal © 2023 IEEE. Personal use of this material is permitted. Permission from
		IEEE must be obtained for all other uses, in any current or future media, including reprinting/republishing this material for advertising or promotional purposes, creating new collective works, for resale or redistribution to servers or lists, or reuse of any copyrighted component of this work in other works. This work is freely available for survey and citation.
		
	\end{minipage}
}
\makeatother
\makeatother
\begin{document}
	\title{An AI-Driven VM Threat Prediction Model for Multi-Risks Analysis-Based Cloud Cybersecurity }
%\title{Online Multiple Risks Analysis based VM Threat Prediction Model for Secure  Data Computation in Cloud Environments}

\author{Deepika~Saxena,~\textit{Member, IEEE,}~Ishu~Gupta,~\textit{Member, IEEE,}~Rishabh~Gupta,\\~Ashutosh~Kumar~Singh,~\textit{Senior Member, IEEE,}~and~Xiaoqing~Wen,~\textit{Fellow, IEEE}
	%and~Jane~Doe,~\IEEEmembership{Life~Fellow,~IEEE}% <-this % stops a space
	\IEEEcompsocitemizethanks{\IEEEcompsocthanksitem D. Saxena and R. Gupta are with Department of Computer Science \& Engineering, The University of Aizu, Japan. E-mail: 13deepikasaxena@gmail.com,  rishabhgpt66@gmail.com.			
		\\  
		I. Gupta is with International Institute of Information Technology (IIIT) Bangalore, India. E-mail:ishugupta23@gmail.com\\
		A. K. Singh is a Director with Indian Institute of Information Technology Bhopal, India and a Professor with the Department of Computer Applications, National Institute of Technology, Kurukshetra, India. E-mail:  ashutosh@nitkkr.ac.in \\		
		Xiaoqing Wen is with the Department of Creative Informatics and the
		Graduate School of Computer Science and Systems Engineering, Kyushu
		Institute of Technology, Fukuoka 8208502, Japan E-mail:
		wen@cse.kyutech.ac.jp
		}}

\maketitle

% make the title area

%\IEEEtitleabstractindextext{
% As a general rule, do not put math, special symbols or citations
% in the abstract or keywords.
\begin{abstract}
	
	Cloud virtualization technology, ingrained with physical resource sharing, prompts cybersecurity threats on users' virtual machines (VM)s due to the presence of inevitable vulnerabilities on the offsite servers. Contrary to the existing works which concentrated on reducing resource sharing and encryption/decryption of data before transfer  for improving cybersecurity which raises computational cost overhead, the proposed model operates diversely for efficiently serving the same purpose. This paper proposes a novel \textbf{M}ultiple \textbf{R}isks {A}nalysis based VM \textbf{T}hreat \textbf{P}rediction \textbf{M}odel (\textbf{MR-TPM}) to  secure computational data and minimize adversary breaches by proactively estimating the VMs threats. It considers multiple cybersecurity risk factors associated with the configuration and management of VMs, along with analysis of users' behaviour. All these threat factors are quantified for the generation of respective risk score values and fed as input into a machine learning based classifier to estimate the probability of threat for each VM. The performance of MR-TPM is evaluated using benchmark Google Cluster and  OpenNebula VM threat traces. The experimental results demonstrate that the proposed model efficiently computes the cybersecurity risks and learns the VM threat patterns from historical and live data samples. The deployment of MR-TPM with existing VM allocation policies reduces cybersecurity threats up to 88.9\%.	  
%In contrast to prior works which engages access control mechanisms, and encryption and decryption of data before transfer, and use of tunnelling for prevention of an unauthorised access to virtual machines which raises excess computational cost overhead, the present invention operates diversely for efficiently serving the same purpose.
\end{abstract}

% Note that keywords are not normally used for peerreview papers.
%\maketitle
%\IEEEdisplaynontitleabstractindextext

% For peer review papers, you can put extra information on the cover
% page as needed:
% \ifCLASSOPTIONpeerreview
% \begin{center} \bfseries EDICS Category: 3-BBND \end{center}
% \fi
%
% For peerreview papers, this IEEEtran command inserts a page break and
% creates the second title. It will be ignored for other modes.
%\IEEEpeerreviewmaketitle

\section{Introduction}
\IEEEPARstart{C}{ybercrimes} are gobbling up the utility of the cloud services for the beneficiaries, including Cloud Service Providers (CSP)s as well as the end users. According to the estimation of Norton Security, in 2023,  cybercriminals  will be breaching 33 billion records per year \cite{cloud2021security}. Also, it has been reported that the misconfiguration and mismanagement associated with the virtualization technology at the cloud platform are the topmost causes of leakage of terabytes of sensitive data of millions of cloud users across the world  \cite{cloud2020security}. Though the CSPs employ sharing of physical resources among multiple users in the view of maximizing the revenues   \cite{saxena2023performance, singh2021quantum, saxena2021op, saxena2022intelligent, gupta2023secom} the  discrepancies and unpatched susceptibilities developed during virtualization, produce misconfigured VMs and hypervisors, expediting the occurrence of cyberattacks. A malicious user may initiate a number of VMs and exploit the misconfigured or vulnerable VMs in multiple ways  to impose a threat on a target VM \cite{9269375, saxena2023sustainable}. Moreover, the vulnerability of hypervisor due to misconfigured virtualization, devastates the cybersecurity by acquiescing all the coresident VMs to be compromised effortlessly \cite{win2017big}. The mismanagement during physical resource distribution yields co-residency of vulnerable VMs and malicious user VM appealing security threats such as leakage of user's confidential data, hampering of data, unauthorized access via insecure interfaces, hijacking of accounts, etc. \cite{han2017using, saxena2021osc, swain2023ai, singh2023bio, saxena2020security}.  Therefore, the key challenge for the CSP is: How to minimize the cybersecurity threats due to misconfiguration and mismanagement of shared resources on a cloud platform?

%The biggest security threats includes \textit{unauthorized access} and \textit{hijacking of accounts} 
%An organizational survey ranked the biggest cloud data security, threats as; \textit{mismanagement of the Cloud platform} (68\%), \textit{unauthorized access} (58\%), \textit{insecure interfaces} (52\%), and \textit{hijacking of accounts} (50\%).
 %The root cause behind these threats is the security loopholes or vulnerabilities associated with the cloud virtualization technology. The sharing of physical resources among multiple users, is embedded virtualization to maximize the revenue for cloud service providers (CSPs) \cite{wu2020data}. However, the  discrepancies and unpatched susceptibilities during virtualization produce misconfigured VMs and hypervisors alongwith mismanagement during VM allocation, opens door to cyberattacks. 
 
 %The co-residency of such vulnerable VMs with malicious user VMs leads to security threats like leakage of cloud user's confidential data \cite{han2017using}, \cite{saxena2020security}, \cite{wu2020secure}. A malicious user may initiate a number of VMs and exploit multiple pathways to impose a threat on target VM \cite{9269375}. The situation becomes more critical with malicious hypervisor where all the coresident VMs can be compromised easily \cite{win2017big}. Therefore, the key challenge for the CSP is how to minimize the data security threats due to misconfiguration and mismanagement of shared resources at cloud platform?
\subsection{ Related Work}
 The considerable works presented for preserving cybersecurity of computational data via VM allocation have focused on both \textit{defensive strategies} as well as \textit{preventive strategies}. The defensive strategies include minimization of resource sharing by reducing the number of users per server \cite{han2017using}, \cite{singh2019secure}, raising the difficulties for achieving co-residency \cite{han2016game}, \cite{juma2018overhead}, and eliminating side-channel based cyberthreats \cite{liang2017mitigating}. While some other researchers have provided \textit{preventive strategies} merely by
  periodic migration of VMs  \cite{saxena2021securevmp}, \cite{yang2018interference}. Levitin et al. \cite{levitin2020co} have presented a method to resist co-residence data theft attacks and improve service reliability by incorporating threshold voting-based N-version programming (NVP). Wu et al. \cite{wu2020secure} presented    a secure and efficient outsourced K-means clustering (SEOKC) scheme for cloud data protection by applying a fully homomorphic encryption with the ciphertext packing technique to attain parallel computation without any excess cost. This scheme preserves data privacy by furnishing database security,  privacy of clustering results, and hidden data   access patterns.  
  Zhang et al. \cite{zhang2019double} presented a double-blind anonymous evaluation-based trust   model which allows suitable matching between  anonymous users and service providers and employed node checking to detect malicious behaviour.   A Previously-Selected-Servers-First (PSSF) policy was proposed in \cite{han2017using} for minimization of exposure of benign VMs to malicious ones. Every server maintained a record of a list of users whose VMs were ever hosted on it. The previously assigned servers that have ever hosted VMs of an old user are considered first for allocation of their new VMs. If no such server exists, then a server with more resource capacity among the remaining servers, is considered for VM placement. Miao et al. \cite{miao2018vm} improved PSSF by adding a rule that a new VM should be co-located with the user VM to whom it is already co-resident. A hierarchical correlation model for analyzing and estimating   reliability, performance, and power
  consumption of a cloud service is proposed in \cite{qiu2015hierarchical} to locate common causes of co-located multilple VM failures sharing multicore CPUs. 
   
 \par
SEA-LB \cite{singh2019secure} allocates VMs considering  minimum power consumption and side-channel attacks with maximum resource utilization by applying modified genetic algorithm approach. The security is provided by minimizing the number of shared servers at the cost of resource utilization. Saxena et al. \cite{saxena2020security} presented a security embedded resource allocation (SEDRA) model in which the performance of network traffic and inter-VM links are considered to detect and mitigate VM threats by utilizing a random tree classifier. Han et al. \cite{han2016game} proposed a two-player game-based defence mechanism against side-channel attack, where the potential differences between the attackers' and legal users'  behaviour were examined by using clustering and semi-supervised learning techniques for the classification of users. As a result, the attacker's efficiency of achieving co-residency with a target VM raised drastically, thus denying an attack on computational data executing within a VM. A data security risk analysis based VM placement is discussed in \cite{han2017reducing}, where a secure and multi-objective VM allocation is formulated and solved by applying an evolutionary optimization. A Vickrey Clarke-Groves bidding mechanism based defence system was presented in \cite{zhang2012incentive} to maximise the difficulty for the adversary to locate the target VM. %Wu et al. \cite{wu2019effective} proposed a cloud resource procurement model to facilitate effective distribution of available cloud resources. This model includes two ways of personalizing reserve prices such as `Lazy' and `Eager' cloud resource procurement to allow safer execution and improved resource utilization. 

\subsection{Our Contributions}
In the light of the aforementioned approaches, it is revealed that  rigorous control over VM-centered cybercrimes is still in the infancy stage which marks the need to proactively estimate the intensity of VM threats in real-time. Since, machine learning algorithms are capable of extracting and learning useful patterns from known malicious activities rapidly by profiling devices such as VMs and servers, and understanding regular activities, it can intelligently identify previously unknown forms of malware and help protect VMs from potential attacks.  Owing to the effective machine-learning capabilities of Extreme Gradient Boosting (XGB) approach including handling missing values, parallelization, distributed computing, and cache optimization, we have devised an XGB inspired VM threat prediction model. Correspondingly, a \textbf{M}ultiple \textbf{R}isks {A}nalysis based VM \textbf{T}hreat \textbf{P}rediction \textbf{M}odel (MR-TPM) is proposed that predicts cyberthreats associated with VMs misconfiguration and their insecure allocation at the cloud platform. To the best of the authors' knowledge, such a proactive VM threat prediction model by considering multiple security risk factors for alleviation of cyberthreats, is presented for the first time. The key contributions are fourfold: 
\begin{itemize}
%	\item  A novel online VM threat prediction model is proposed for secure resource management in cloud environment. 
	\item A novel concept of multiple risks analysis based  cybersecurity pertaining to VMs,  is proposed to maximize the security of computational data executing on VMs. Also, the  ill-effects of misconfiguration and insecure VM management are minimized by considering the intended multiple risk factors.%Also, a proactive VM threat prediction and mitigation using a newly introduced Risk-Score Matrix is proposed for multiple risks analysis before threat prediction. 
	%\item A dynamic ensemble workload predictor is developed to forecast resource usage with enhanced accuracy which triggers VM migration to alleviate effect of over/under-load on server before its actual occurence and improve performance of datacenter.
	\item Quantification and assessment of all the considered security threat factors for the periodic training of the newly developed artificial intelligence (AI) driven VM threat prediction model is introduced.
	\item Implementation and evaluation of the proposed model using real VM threat traces  reveals that MR-TPM predicts threats with precise accuracy and helps to mitigate them before the occurrence.% hence, reduces significant number of VM threats as compared to existing methods.  
	\item Deployment of the proposed model with existing VM placement policies demonstrates its compatability and applicability in improving the security of user data during execution by exploiting and analysing multiple VM risks for threat prediction. Additionally, its workload prediction component helps to optimize resource utilization, power consumption substantially by minimizing the number of active servers.

\end{itemize}
% The very first letter is a 2 line initial drop letter followed
% by the rest of the first word in caps.
% 
% form to use if the first word consists of a single letter:
% \IEEEPARstart{A}{demo} file is ....
% 
% form to use if you need the single drop letter followed by
% normal text (unknown if ever used by the IEEE):
% \IEEEPARstart{A}{}demo file is ....
% 
% Some journals put the first two words in caps:
% \IEEEPARstart{T}{his demo} file is ....
% 
% Here we have the typical use of a "T" for an initial drop letter
% and "HIS" in caps to complete the first word.

 A bird eye view of the proposed model is shown in Fig. \ref{fig:bird-eye-view}, where multiple types of VM security attack factors ($\{{R}_1, {R}_2, ..., {R}_n\} \in {R}$) are gathered, quantified, and analysed to periodically train a machine learning based VM threat estimator for accurate prediction of future threats on VMs. 
\begin{figure}[!htbp]
	\centering
	\includegraphics[width=0.95\linewidth]{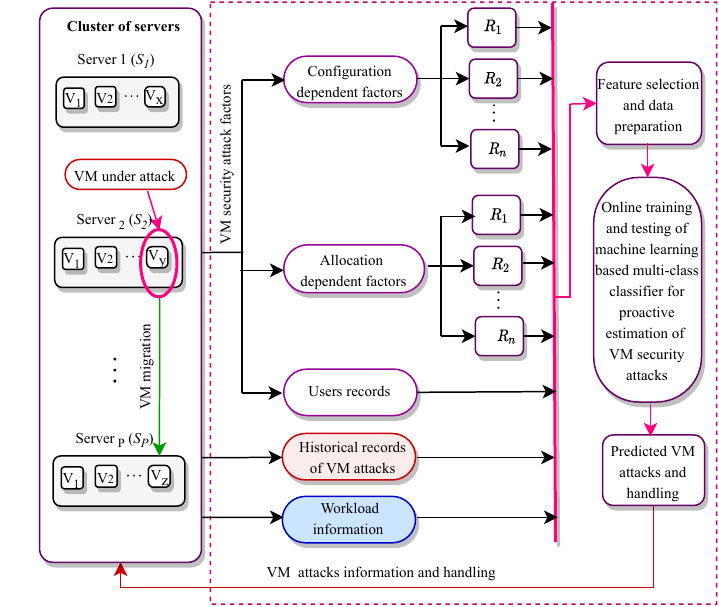}
	\caption{A bird eye view of MR-TPM}
	\label{fig:bird-eye-view}
\end{figure}

\textit{Organization}: The paper is structured as follows: Section II discusses the problem formulation. A detailed elaboration of proposed MR-TPM is conferred in Section III. The multiple cyberthreat factors associated with VMs, including user behaviour analysis, configuration-dependent factors, and allocation-dependent factors, are entailed in Section IV, Section V, and Section VI, respectively. The operational design and complexity of MR-TPM is presented in Section VII. The performance evaluation followed by conclusive remarks and the future scope of the proposed work are presented in Section VIII and  Section IX, respectively. Table \ref{table:notation} shows the list of symbols with explanatory terms used throughout the paper.
\begin{table}[!htbp]
	\centering 	
	\caption[Table caption text] {Notations with explanatory terms}  %\cite[p.10]{refid} }
	\label{table:notation}
	%\resizebox{8cm}{!}{
	\resizebox{0.970\textwidth}{!}{\begin{minipage}{\textwidth}
			\begin{tabular}{|l|}
				\hline 
				%\multicolumn{2}{c}{Item} \\
				%\cline{1-2}
			%	Symbols: Explaination terms  \\
		%	&	\hline \hline
				%$v$    &virtual machine \\
				%$t$          & task      \\
				%$S$      & server \\
				
				$S$: server, $V$: VM, $U$: user, 
				$P$: number of servers, $Q$ : number of VMs,\\ $M$: number of users,
				$\omega$: mapping among server, VM, user \\
				%$\Omega$&penalty cost\\
				${R}$: security risk factor,
				${L}$:VM vulnerability,  ${H}$: Hypervisor vulnerability\\
				${S}^{Hyp\_scor}$: Server's hypervisor own vulnerability, $\Xi$: Threat,\\
				${C}$: Co-residency effects, ${N}$: Network cascading effects,\\
				
				$ W^p$: predicted workload, 
				$\mathds{F}$: features used for prediction,\\
				$\mathds{BL}$: Base Learners in XGB,
				$\Theta$:unauthorized access,
				%\vartheta$ & Sum of network connections\\
				$Rq$: job request,\\
				$\mathds{H^\ddagger}$: record of malicious actions,
				%	$p$,$q$, $m$& number of severs, VMs, users\\
				$\mathds{P}(\Xi)$: probability of threat \\
				%	$\psi$ & VM allocation\\
				$RU$: resource utilization, $PW$: power consumption,
				$\mathchorus{M_{c}}$: migration cost\\
				$\mathds{CC}$: status of VM after migration, $\mathds{G}$:status of server after migration\\
					$\mathds{D}(S_k, S_j)$: distance between servers $S_k$ and $S_j$\\
			
				\hline
			\end{tabular}
			%}
	\end{minipage}}
\end{table} 
\section{Problem Formulation}
A cloud datacenter environment is considered where  multiple users requests for execution of their workloads or applications. The users can be categorised into legitimate (normal) and malicious (threat-imposing) users. During workload execution, the inter-dependent VMs need to communicate and exchange information to complete the application execution. However, some malicious user VMs may intrude this operation and seek for the security loopholes to  exploit various opportunities for launching  successful threats to legitimate users' VMs  for stealing sensitive information via an unauthorized access. The security of VMs is compromised by exploiting either configuration discrepencies of VMs and associated host servers or insecure allocation and mismanagement of VMs. Accordingly, a problem configuring research assumptions and design goals is formulated in the following subsections. 
\subsection{Assumptions}
The assumptions addressing conditions for VM threats and the capabilities of malicious user VMs during workload distribution and execution are as follows:

\begin{itemize}
	\item Only CSP decides mapping between VMs and servers, and it may or may not have the knowledge of legitimate and malicious VMs.
	
	%	\item The software components related to VM placement and migration processes are all secure.
	\item Each active VM belongs to one user only. However, the user can have multiple number of VMs over time. 
	
	\item Malicious user may run one or multiple VMs to exploit means of security escape for imposing a threat on target VM(s). The VM threats can be executed in three ways: one-to-one (one specific malicious VM attacks one target VM), one-to-many (one specific malicious VM attacks multiple target VMs in networking), and many-to-many (group of malicious VMs attack many target VMs).  
	 
	\item  VM(s) are migrated either to handle over/under-load on the source server or to protect them from malicious activity only. Otherwise, the VM is assumed to run on the same server until the user terminates it.
\end{itemize}
\subsection{Problem statement and Design Goals}
Specifically, the problem is to develop a VM threat prediction model which is trained with data samples  considering $n$ probable risk factors addressing security loopholes that estimates VM(s) security threats proactively to improve cybersecurity during cloud workload processing. Based on the aforementioned problem assumptions and statement,  the design goals of the proposed model are as follows:
\begin{itemize}
		\item To develop a machine learning-driven model that will determine VM threats prior to occurrence in real-time. This model must not effect the efficiency of VM management and it must be adaptable and compatable for operation with any VM allocation policies. 
 \item To generate a knowledge database for training of the corresponding VM threat predictor by identifying and computing risk score values for all the probable security factors associated with VM(s). 
 	\item   To accurately detect security threats on legitimate VM(s)  due to presence of malicious VM and vulnerabilities of VM(s) configuration and management. 
\end{itemize}

\section{Proposed VM Threat Prediction Model}
\begin{figure*}[!htbp]
	\centering
	\includegraphics[width=0.7\linewidth]{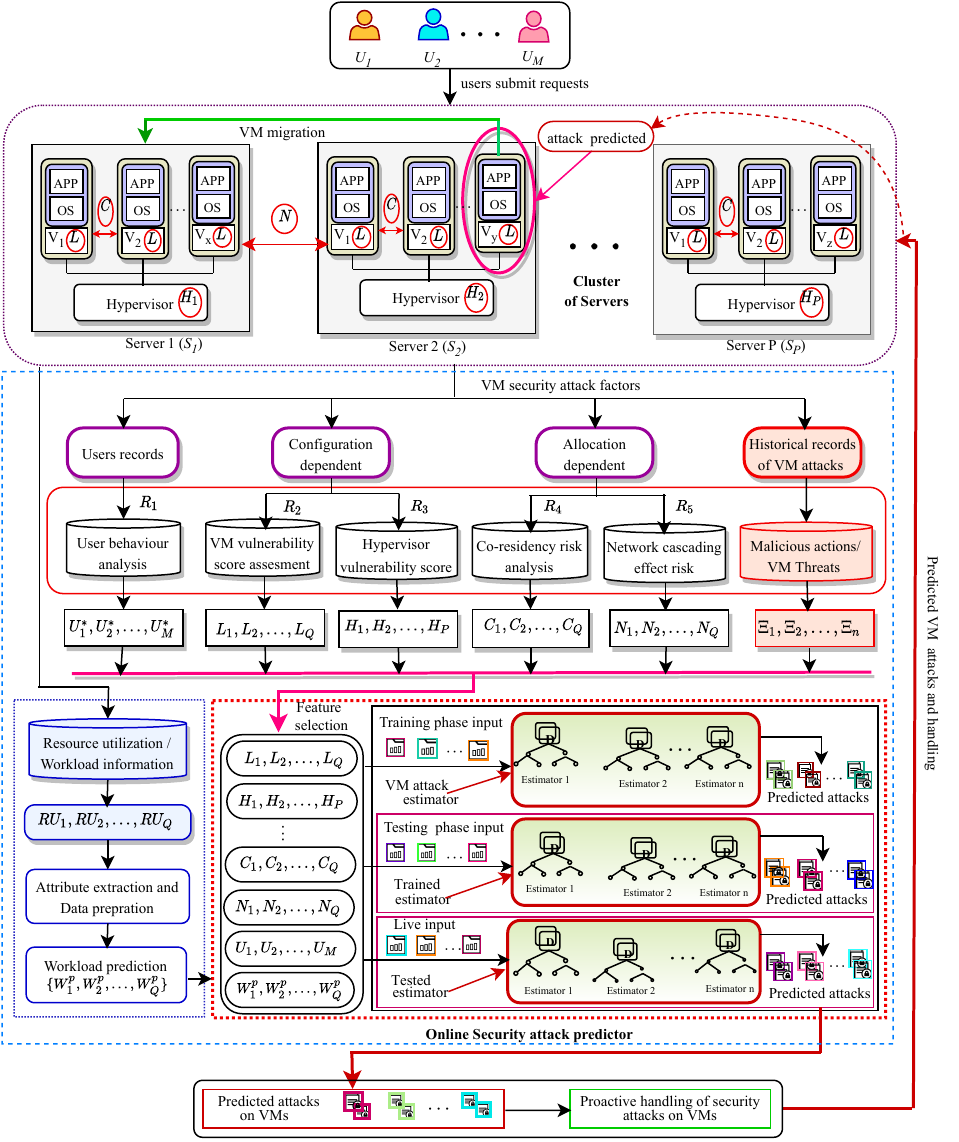}
	\caption{  {M}ultiple {R}isks {A}nalysis based VM {T}hreat {P}rediction {M}odel (MR-TPM)}
	\label{fig:proposed-model}
\end{figure*}
Consider a {cluster} of $P$ {servers} $\{S_1, S_2, ..., S_P\}\in\mathds{S}$ hosts $Q$ {VMs} $\{V_1, V_2, ..., V_Q\} \in\mathds{V}$ of $M$ {users} $\{U_1, U_2, ..., U_M\}\in\mathds{U}$. Let $S_1$ hosts $x$ VMs such that $\{V^1_1, V^1_2, ..., V^1_x\} \in\mathds{V}^1$; $S_2$ and $S_P$ host $y$ VMs $\{V^2_1, V^2_2, ..., V^2_y\} \in\mathds{V}^2$ and $z$ VMs $\{V^P_1, V^P_2, ..., V^P_z\} \in\mathds{V^P}$, respectively, where $\{\mathds{V}^1, \mathds{V}^2, ..., \mathds{V}^P\} \in\mathds{V}$ and $\{x\cup y\cup z\}\subseteq Q$. A mapping $\omega | \omega$ :$\mathds{U} \times \mathds{V}\mapsto\mathds{S}$ assigns VMs of each user on a specific server such as $\omega_{ji}^k =1$ iff $j^{th}$ VM of $k^{th}$ user is deployed on $i^{th}$ server. The comprehensive description of the essential blocks and intrinsic information flow of MR-TPM is depicted in Fig \ref{fig:proposed-model}.   The proposed cyberthreat prediction model records and analyses multiple { security risk factors} $\{{R}_1, {R}_2, {R}_3, {R}_4, {R}_5\} \in {R}$ associated with a VM {configuration}, such as {VM vulnerability} \{${L}_1, {L}_2, ..., {L}_Q$\}, {server Hypervisor vulnerability} \{${H}_1, {H}_2, ..., {H}_P$\}; and {VM allocation}, including {Side-channel effect} \{${C}_1, {C}_2, ..., {C}_Q$\} and {Network cascading effect} \{${N}_1, {N}_2, ..., {N}_Q$\}; {User behaviour} \{$U^{\ast}_1, U^{\ast}_2, ..., U^{\ast}_M$\}; and previous records of {VM threats} \{$\Xi_1, \Xi_2, ..., \Xi_n$\}. During time-interval \{$t_a, t_b\}\in t$, all the security risk factors and threat information are collected and categorized into four {classes}:
\begin{itemize}
	\item User behaviour analysis \{$U^{\ast}_1, U^{\ast}_2, ..., U^{\ast}_M$\} (Section \ref{userba})
	\item { VM Configuration-dependent factors} for computation of the scores of VM vulnerability \{$L_1, L_2, ..., L_Q$\} and server hypervisor vulnerability scores \{${H}_1, {H}_2, ..., {H}_P$\} (Section \ref{config})
	\item {VM Allocation-dependent factors} for assessment of side-channel effects \{${C}_1, {C}_2, ..., {C}_Q$\} and network cascading effects \{${N}_1, {N}_2, ..., {N}_Q$\} (Section \ref{allocation})
	\item  {Records of live threats or malicious actions on VMs} for updation of VM threat database 
\end{itemize}
 %($i$) { VM usage traces } by $\mathds{U}$ to analyse user behaviour , ($ii$) { VM Configuration dependent factors} for computation of the scores of VM vulnerability \{$L_1, L_2, ..., L_Q$\} and server hypervisor vulnerability scores \{${H}_1, {H}_2, ..., {H}_P$\}, ($iii$) {VM Allocation dependent factors} for assesment of side-channel effects \{${C}_1, {C}_2, ..., {C}_Q$\} and network cascading effects \{${N}_1, {N}_2, ..., {N}_Q$\} and ($iv$) { Live threat records or malicious actions on VMs} for updation of VM threat database.    

%\begin{definition}

%\end{definition}
%\begin{definition}
	
%\end{definition}

\begin{definition}
	\textbf{VM cyberthreat prediction}: The mechanism intended for computation and analysis of various security escapes and unpatched discrepancies associated with a VM along with proactive threat estimation,  is designated as VM cyberthreat prediction.
\end{definition}
 MR-TPM proactively estimates the {workload information} \{$W^p_1, W^p_2, ..., W^p_Q \in \mathds{W^p}$\} by utilizing a {neural network based workload predictor} ($Pr$), which is periodically trained with the latest and historic {resource utilization} \{$RU_1, RU_2, ..., RU_Q$\} by VMs \{$V_1, V_2, ..., V_Q$\} to determine {active VMs} \{$\hat{V}_1, \hat{V}_2, ..., \hat{V}_{Q^{\ast}}$, $Q^{\ast}\subseteq Q $\} having {predicted workload} ($W^p$) $> 0$. The prior knowledge of active VMs is utilized for analysis of the placement of VMs during the next time interval \{$t_{a+1}, t_{b+1}\} \in t$. The consecutive processes of {feature selection} ($\mathds{FS}$) and {threat prediction} ($\mathds{TP}$) is performed for active VMs based on the predicted workload information \{$W^p_1,W^p_2, ..., W^p_Q$\} for VMs \{$V_1, V_2, ..., V_Q$\} of users \{$U_1, U_2, ..., U_M$\} during time-interval \{$t_{a+1}, t_{b+1}$\}. The historical database of VM threats is utilized for feature selection, followed by training of online VM threat predictor $\mathds{TP}$, which is periodically re-trained with the latest data samples for online VM threat prediction \{$\Xi^p_1, \Xi^p_2, ..., \Xi^p_{Q^{\ast}}$\}.
 Among all the collected and analysed features \{$ L, {H}, {C}, {N}, \mathds{V}, \mathds{S}, \mathds{U}, \mathds{U}^{\ast}, W^p, \omega $ etc.\} $\subseteq \mathds{F}$, only useful features are filtered (i.e., $\mathds{F}^{\ast}$) by applying {Recursive Feature Elimination} (RFE) to train threat predictor $\mathds{TP}$ to estimate VM threats \{$\Xi^p_1, \Xi^p_2, ..., \Xi^p_{Q^{\ast}}$\} with improved accuracy and reduced computation time. The data samples containing selected features \{$\mathds{F}_1^{\ast}, \mathds{F}_2^{\ast}, ..., \mathds{F}_{s}^{\ast}$\} $ \in\mathds{F}^{\ast}$ are split into {training samples} \{$\bar{\mathds{F}}_1^{\ast\ast}, \bar{\mathds{F}}_2^{\ast\ast}, ..., \bar{\mathds{F}}_{s^{\ast}}^{\ast\ast}$\} $ \in \bar{\mathds{F}}^{\ast\ast}$  and {testing samples} \{$\mathds{F}_1^{\ast\ast}, \mathds{F}_2^{\ast\ast}, ..., \mathds{F}_{t^{\ast\ast}}^{\ast\ast}$\} $ \in\mathds{F}^{\ast\ast}$ subject to the constraints: (\textit{i}) $ \mathds{F}^{\ast}=\bar{\mathds{F}}^{\ast\ast} \cup {\mathds{F}}^{\ast\ast}$ (\textit{ii}) $\bar{\mathds{F}}^{\ast\ast} \cap {\mathds{F}}^{\ast\ast}= \emptyset$ (\textit{iii}) $\{s^{\ast}, s^{\ast\ast}\} \subseteq s$ where $s$ is total number of data samples. A mapping \{$\Omega|\Omega:\bar{\mathds{F}}^{\ast\ast} \times \mathds{TP} \Rightarrow \mathds{TP}^{\ast} $\} trains threat predictor $\mathds{TP}$ with $\bar{\mathds{F}}^{\ast\ast} $ to generate {Trained Predictor} ($\mathds{TP}^{\ast}$) during training phase while \{$\Omega^\ast|\Omega^\ast:{\mathds{F}}^{\ast\ast} \times \mathds{TP}^{\ast} \Rightarrow \mathds{TP}^{\ast\ast} $\} evaluates $\mathds{TP}^{\ast}$ with unseen test data ${\mathds{F}}^{\ast\ast} $ to generate {Tested Predictor} ($\mathds{TP}^{\ast\ast}$) for online VM threat prediction.
\par The proposed VM threat predictor utilizes an {Extreme-Gradient Boosting} (XGB) based machine learning algorithm to learn and develop the precise correlations among extracted patterns for accurate prediction of cyberthreats: \{$\Xi^p_1$, $\Xi^p_2$, ..., $\Xi^p_{Q^{\ast}}$\}. Let a threat predictor ($\mathds{TP}$) is composed of $l$ base learners (i.e., decision trees) $ \mathds{BL}^\ast=\{\mathds{BL}^\ast_1, \mathds{BL}^\ast_2, ..., \mathds{BL}^\ast_l\}$ and predicts output $\mathds{O}^\ast=\{\mathds{O}^\ast_1, \mathds{O}^\ast_2, ..., \mathds{O}^\ast_l\}$ using Eq. (\ref{eq:xgb1}), where $\mathds{F}_i \subseteq \mathds{F}^{\ast}$ such that $\mathds{F}$ represents the input vector of size $s^\ast$.  During each iteration, decision trees are trained incrementally to reduce prediction errors and the amount of error reduction  is computed as \textit{gain} or  \textit{loss term} using Eq. (\ref{eq:xgb2}). The expressions $L(\mathds{O}^\ast, {\mathds{O}^{\ast\ast}_{t-1}} + {\mathds{BL}^\ast_t}(\mathds{F}_i))$ and $\Psi(\mathds{BL}^\ast_t)$ are {loss term} and a {regularization term}, respectively. Taylor expansion is applied to compute the exact loss for different possible base learners, which updates Eq. (\ref{eq:xgb2}) to Eq. (\ref{eq:xgb3}); where $g_i=\partial_{\mathds{O}^{\ast\ast}_{t-1}}{L(\mathds{O}^\ast, {\mathds{O}^{\ast\ast}_{t-1}})}$, and $h_i=\partial^2_{\mathds{O}^{\ast\ast}_{t-1}}{L(\mathds{O}^\ast, {\mathds{O}^{\ast\ast}_{t-1}})}$ are first and second order derivatives of loss function in the gradient, respectively. The term $	\Psi(\mathds{BL}^\ast_t)$ is computed using Eq. (\ref{eq:xgb5}),
 where $\gamma$ and $\lambda$ are $L_1$ and $L_2$ regularisation coefficients, respectively, $w$ is internal split tree weight and  $K$ is the number of leaves in the tree.
% $\mathds{BL}^\ast_t$ has $K$ leaf nodes $ I_j $ is a set of instances for $j^{th}$ node and its prediction is $ w_j $.
\begin{gather}
\label{eq:xgb1}
	\mathds{O}^\ast=\sum_{z=1}^{l}\mathds{BL}_z^\ast(\mathds{F}_i) \quad \forall i \in \{1, 2, ..., s^\ast\}
\\
	\label{eq:xgb2}
	{L}_t = \sum_{i=1}^{ s^\ast} {L(\mathds{O}^\ast, {\mathds{O}^{\ast\ast}_{t-1}} + {\mathds{BL}^\ast_t}(\mathds{F}_i)) + \Psi(\mathds{BL}^\ast_t)}
\\
  \label{eq:xgb3}
  {L}_t = \sum_{i=1}^{ s^\ast} { [g_i\mathds{BL}^\ast_t(\mathds{F}_i) + \frac{1}{2}h_i {\mathds{BL}^\ast_t}(\mathds{F}_i)] + \Psi(\mathds{BL}^\ast_t)}  
 \\
  \label{eq:xgb5}
 	\Psi(\mathds{BL}^\ast_t)=\gamma K + \dfrac{1}{2} \lambda ||w||^2
\end{gather}   
During each time-interval $\{t_a, t_b\}\in t, a<b$, live selected features  $\hat{\mathds{F}}^{\ast\ast} $ are given as input to the above discussed threat predictor $\mathds{TP}^{\ast\ast} $ to estimate the status of threat $\Xi$ for  VMs \{$\hat{V}_1, \hat{V}_2, ..., \hat{V}_Q^{\ast}$\} in the next time-interval  $\{t_{a+1}, t_{b+1}\}\in t, a<b$. Accordingly, the process of {VM-threat handling} is performed for the VMs with predicted threat-status ($\hat{V}_i^{\Xi} >0: i \in [1, 2, ... Q^{\ast\ast}], Q^{\ast\ast} \subseteq Q^{\ast}\subseteq Q $) by shifting them to a server where the possibility of threat is least ($\hat{V}_i^{\Xi} = 0$). A detailed description of VM security risk factors is provided in the subsequent sections.

\section{User behaviour analysis} \label{userba}
Users \{$U_1, U_2, ..., U_M$\} submit {job requests} \{$Rq_1, Rq_2, ..., Rq_M$\} during time-interval $\{t_a, t_b\}$ at the cloud platform as depicted in Fig. \ref{fig:userbehaviour}, where the users are classified into {Trusted}, {Non-trusted} and {Unknown} users. The $k^{th}$ user $U_k$ behaviour is defined in accordance with the actions of its VMs as follows: 	\textit{Trusted}: The user behaviour is trusted when the VMs of known user $U_k$ (having historical  records of VM resource usage), execute assigned load efficiently without interrupting and interfering with other co-located VMs via an unauthorised access, irrespective of the presence of any vulnerabilities of software or hardware. \textit{Non-trusted}: A user $U_k$  is non-trusted in case of the users VM attempt any kind of cybercrime or malicious activity such as unauthorized data access, data hijacking, data phishing, etc. by leveraging the susceptibilities of cloud virtualization technology. \textit{Unknown}: The new user for which there are no records of any  previous VM usage, is considered as unknown user.
{User behaviour analysis deals with the process of critical monitoring, recording, and examining the traces of their previous VM usage and the inter-relationships among co-resident VMs of different users  periodically to interpret or investigate the occurrence of cyberthreats in the presence of intended vulnerabilities of cloud environments.} The class of user and the selected VM placement policy are passed to the load balancer, which makes VM management decisions. Accordingly, the VMs are deployed on different servers to compute the users' data \{$Rq_1, Rq_2, ..., Rq_M$\}. Concurrently, the VM usage traces or type of data access information is collected and passed to `VM usage database' for examination of the user behaviour. 
	\begin{definition}
		\textbf{VM usage database} ($\mathds{DB}$): The historical repository of data values concerning VM usage related attributes such as its ephemeral user ID, CPU, memory, and bandwidth usage, inter-communication links with other VMs, types of authorized access, etc., constitute  VM usage database which is utilized for multiple risks computation, training of resource usage predictor, and VM threat predictor.
	\end{definition}  
\begin{figure}[!htbp]
	\centering
	\includegraphics[width=0.95\linewidth]{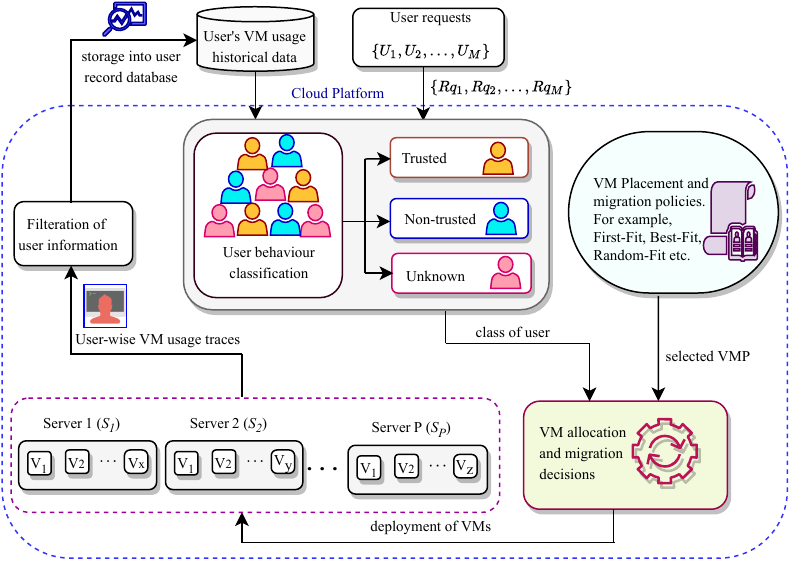}
	\caption{User classification}
	\label{fig:userbehaviour}
\end{figure}  
The new VM of $k^{th}$ user ($U_k$) is allocated according to the analysis of $U_k$ behaviour by applying Eq. (\ref{uba}), where $\Theta_k$ represents malicious actions for e.g., {unauthorized access} executed by $U_k$.    
\begin{equation}
\label{uba}
U_k= \begin{cases}
Trusted (0), & {If( \Theta_k=0)} \\
Non-trusted (1), & {If(\Theta_k>0)} \\
Unknown (-1), & {\text{otherwise}}  
\end{cases}	
\end{equation}

\begin{theorem}
	The behaviour of user $U_k^\ast$ having VM $V_i^{k^\ast}$ hosted on server $S_j$ is bounded by $\Theta$ such that for a given time-period \{$t_a$, $t_b$\} and VM usage database ($\mathds{DB} \neq \phi$), if $\Theta_k^\ast= 1$, $U_k^\ast$ is non-trusted; otherwise, it is trusted.
\end{theorem}
\begin{proof}
	Let $\Theta_{ij,k \Rightarrow {i^\ast}j,{k^\ast}}$ represents a data access by a user $U_k^\ast$ owning VM $V_{i^\ast}^{k^\ast}$ to $k^{th}$ user $U_k$ VM $V_{i}^k$ during time \{$t_a$, $t_b$\}, is formulated in Eq. (\ref{threat1}): 
	\begin{equation}
	\label{threat1}
	\int_{t_a}^{t_b} \Theta_{ij,k \Rightarrow {i^\ast}j,{k^\ast} } dt= {\int_{t_a}^{t_b} (\omega_{ij}^k \times \omega_{{i^\ast}j}^{k^\ast}) \times \uplus^v_{ij,k \Rightarrow {i^\ast}j,{k^\ast}} } dt
	\end{equation}
	where, $\uplus^v_{ij,k \Rightarrow {i^\ast}j,{k^\ast}}$ represents inter-VM relationship between $V_{i}^k$ and $V_{i^\ast}^{k^\ast}$, $\forall \{i, {i^\ast}\} \in Q, j \in P$. It is a Boolean value, $1$ for unauthorised data access (i.e., non-trusty relation) and $0$ for trusty relation. Assume $\mathds{LA}$ (stated in Eq. (\ref{links}))  specifies set of authorized inter-VM links for $i^{th}$ VM of ${k^\ast}^{th}$ user. The  inter-VM relationship ($\uplus^v_{ij,k\rightarrow i^{\ast}j,k^{\ast}}$) between $V_{i}^k$ and $V_{i^{\ast}}^{k^\ast}$ placed on $j^{th}$ server is evaluated in Eq. (\ref{linkthreat}) which corresponds to the inter-VM links ${\mathds{LA}_{ij,k \Rightarrow {i^\ast}j,{k^\ast}}}$ between them. \\
	\begin{gather}
	\mathds{LA}^{V_{i,k^\ast}} \in \{\mathds{LA}_1^{V_{i,k^\ast}}, \mathds{LA}_2^{V_{i,k^\ast}}, ..., \mathds{LA}_n^{V_{i,k^\ast}}\} \label{links}  \\
		\uplus^v_{{ij,k\rightarrow i^{\ast}j,k^{\ast}}} = 	\begin{cases}
	$1$, & {If \quad {\mathds{LA}_{ij,k \rightarrow  i^{\ast}j,k^{\ast}}}} \nsubseteq \mathds{LA}^{v_{i,k}}  \\
	$ 0$, & { {\textit{otherwise}} } 
	\end{cases} \label{linkthreat}
	\end{gather}
    
 Hence, when the user $U_k^\ast$ has attempted an unauthorised access, the inter-VM relationship parameter $\uplus^v_{{ij,k\rightarrow i^{\ast}j,k^{\ast}}}$ is equal to $1$ and applying Eq. (\ref{linkthreat}) in Eq. (\ref{threat1}), $\Theta_k^\ast= 1$ for $U_k^\ast$ is non-trusted, and trusted, otherwise. 	
\end{proof}

 \begin{corollary}
 	The user $U_k^\ast$ behaviour is also reflected by the relationship $\uplus^S_{{ij,k} \rightarrow S_{j}}$ between user $U_k^\ast$ and  server $S_{j}$ which is `non-trusty' for malicious records ($\mathds{H^\ddagger}_j$) greater than $0$, otherwise, it is trusty.
 \end{corollary}

\begin{proof}
	Let an unauthorized data access $\Theta_{{ij,k^\ast} \Rightarrow S_{j} }$ from $i^{th}$ VM of ${k^\ast}^{th}$ user to server $S_{j}$ during time $\{t_a, t_b\}$ is formulated in Eq. (\ref{threat2}). The term $ \uplus^S_{ij,k^\ast}= \{0, 1\}$ signifies a relationship between $S_{j}$ and $U_k^\ast$, such that it is equals to a Boolean value, $1$ for an unauthorized data access via malicious hypervisor, and $0$ otherwise.
	\begin{gather} \label{threat2}
	\int_{t_a}^{t_b} \Theta_{{ij, k} \Rightarrow S_{j} }dt = {\int_{t_a}^{t_b} \omega_{ij,k} \times \uplus^S_{{ij,k} \rightarrow S_{j}}}dt \quad \forall \{i\} \in Q, j \in P 
	\end{gather}
	Suppose the relation ($\uplus^S_{{ij,k} \rightarrow S_{j}}$) between user $U_k$ and  server $S_{j}$ is analysed using Eq. (\ref{thresholdthreat}), where $\mathds{H^\ddagger}_j$ represents {malicious actions records} computed using Eq. (\ref{eq:serverthreat}). 
	\begin{gather}
	\uplus^S_{{ij,k} \rightarrow S_{j}} = 	\begin{cases}
	$1$, & {If(\mathds{H^\ddagger}_j >0 \land {H_j} > {H}_{Thr}) } \\
	$ 0$, & { {\textit{otherwise}} } 
	\end{cases} \label{thresholdthreat}\\
	\label{eq:serverthreat}
	\mathds{H^\ddagger}_j = \sum{\omega_{ij}^k \times \omega_{{i^\ast}j}^{k^\ast} \times \Theta_{ij,k \Rightarrow {i^\ast}j,{k^\ast} } } 
	\end{gather}
	If user $U^{k^\ast}$ is non-trusty, then $\Theta_{ij,k \Rightarrow {i^\ast}j,{k^\ast}} = 1$  (as proved in Theorem 1). Accordingly, the value of the term $\mathds{H^\ddagger}_j$ is also $1$. Putting $\mathds{H^\ddagger}_j$= $1$ in Eq. (\ref{thresholdthreat}) when ${H_j} > {H}_{Thr}$, the value of $\uplus^S_{{ij,k} \rightarrow S_{j}}$ becomes $1$. Hence, $\mathds{H^\ddagger}_j > 0$ for a non-trusty behaviour of user $U^{k^\ast}$.  
		
\end{proof}
Further, the total threat information or unauthorized data access $\Theta_k$  for the duration \{$t_a$, $t_b$\} by $U^k$ is compiled by applying Eq. (\ref{threat3}):
\begin{equation}
\label{threat3}
\int_{t_a}^{t_b} \Theta_{k }dt = {\int_{t_a}^{t_b} (\Theta_{ij,k \Rightarrow {i^\ast}j,{k^\ast} } + \Theta_{{ij, k} \Rightarrow S_{j} })}dt
\end{equation}

The {Random Forest Classifier} (RFC) classifies users $U_1, U_2, ..., U_M$ on the basis of their future behaviour  by utilizing the learning capability of different base learners or decision trees and knowledge driven via extracted correlated patterns from their historical information, where Eq. (\ref{uba}) is evaluated periodically for duration \{$t_a, t_b$\} $\in t$. RFC is composed of $n^\ast$ {base learner estimators} that produce ${n^\ast}$ outcomes and apply majority voting to predict absolute behaviour of user $U_k$.

%The four dimensional risk analysis is carried out for each LAVN and various risk scores are produced respective to each LAVN. The risk scores and threat database information is utilized by the threat predictor to estimate/detect threat on any VM of LAVN and tries 

%AssumeVMsofcommonuserandotherVMs(authorizedbyrespective VM owner) are interdependent and determines legal communication links between VMs, except which all other links are gateway to illegal data access. 
 %accompainied with a suitable VM placement policy to assign the VMs $\mathchorus{v}_1$, $\mathchorus{v}_2$, ..., $\mathchorus{v}_Q$ of users $\mathchorus{u}_1$, $\mathchorus{u}_2$, ..., $\mathchorus{u}_M$ to the selected servers among $\mathchorus{S}_1$, $\mathchorus{S}_2$, ..., $\mathchorus{S}_P$ 

\section{Configuration-dependent factors} \label{config}
 The  vulnerabilities of virtualisation technology and VM security loopholes which are governed by the susceptibilities related to the creation and installation of VMs, including sharing of a common physical machine, hypervisor or guest OS installation, are confined to configuration-dependent risks. MR-TPM considers two {configuration dependent} security risk factors (${R_2}, {R_3}$), including {VM vulnerability} ($L$) \cite{chiang2014swiper} and {Hypervisor vulnerability} (${H}$) \cite{holm2012empirical}. A {malicious user} ($U^{Mal}: U^{Mal} \subseteq \mathds{U}$) launches multiple applications ($A_p$, $A_q$, ..., $A_t$) to compromise the {target benign VM} ($V^{Ben}: V^{Ben} \subseteq \mathds{V}$) by achieving co-residency and exploiting VM and hypervisor vulnerabilities, as shown in Fig. \ref{fig:VM_Hyper_Vul}. The application $A_p$ of $U^{Mal}$ exploits the hypervisor vulnerability of server $S_1$  (i.e., $H_1 > H_{Thr}$) and compromises multiple VMs. At server $S_p$, the applications $A_s$ and $A_t$ of $U^{Mal}$ utilize the vulnerability of VM $V_2$ (i.e., $L_2 > L_{Thr}$) to launch the attack and hamper computational data on it. The parameters $H_{Thr}$ and $L_{Thr}$ specify threshold values of hypervisor vulnerability  and VM vulnerability, respectively. At server $S_2$, both kinds of vulnerabilities are absent, i.e., the threshold values of VM vulnerability as well as hypervisor vulnerability are lesser than their respective threshold values, and all the VMs deployed on it are secured ($V^{\Xi}=0$). 
 \begin{figure*}[!htbp]
 	\centering
 	\includegraphics[width=0.65\linewidth]{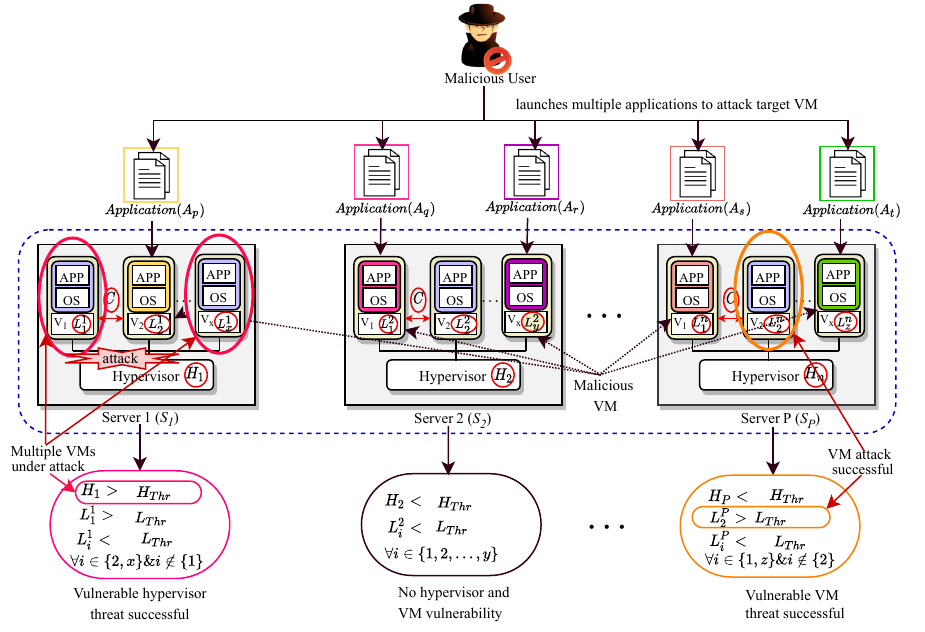}
 	\caption{VM and hypervisor vulnerability based threats}
 	\label{fig:VM_Hyper_Vul}
 \end{figure*}
\par The {vulnerable VMs} are deprived of standard security features with respect to the operating system, applications like e-mail, web-browsing, and network protocols, and are prone to loose administrative control. Besides this, {vulnerable hypervisors of servers} leads to {hyperjacking} where $U^{Mal}$ can easily gain unauthorized access of hypervisor to compromise all the hosted VMs and the applications running on them. It is typically launched against {Type 2-Hypervisors} running over a host operating system. 
A mapping \{$\varpi|\varpi:A_m \times U^{Mal}\times V_{i} \Rightarrow V^{Mal}_i$\} defines {malicious VMs} such that an $i^{th}$ VM ($V_{i}$) becomes malicious, if it hosts $m^{th}$ application ($A_m$) of $U^{Mal}$. The probability of threat ($\mathds{P}(\Xi_i)$) for $i^{th}$ VM over time-interval \{$t_a, t_b$\} can be defined using Eq. (\ref{eq_VM_Hypr_threat}),
\begin{equation}
\label{eq_VM_Hypr_threat}
	\mathds{P}(\Xi_i)= \begin{cases}
1, & {If( L_i > L_{Thr} \quad \&\& \quad \omega_{ji}^k \cap\omega_{ji^{\ast}}^{k^{\ast}} = S_j)} \\
1, & {If({H}_j > {H}_{Thr} \quad \&\& \quad \omega_{ji}^k \cap \omega_{ji^{\ast}}^{k^{\ast}} = S_j)} \\
0, & {\text{otherwise}}  
\end{cases}	
\end{equation}
where $t_a < t_b$ and $\omega_{ji}^k \cap\omega_{ji^{\ast}}^{k^{\ast}} = S_j$ signifies co-location of $i^{th}$ VM ($V_{i}$) of $k^{th}$ {benign user} ($U^{Ben}|U^{Ben} \subseteq \mathds{U}$) and ${i^{\ast}}^{th}$ malicious VM ($V_{i^{\ast}}\subseteq V^{Mal}$) of ${k^{\ast}}^{th}$ malicious user $U^{Mal}$ at  $j^{th}$ server ($S_{j}$). 

\subsection{VM vulnerability}
The VMs vulnerability score list is generated using vulnerability scanner tools, such as Common Vulnerability Scoring System (CVSS), Nessus and Qualys \cite{holm2012empirical}. The CVSS  measures the severity of vulnerabilities  of a hardware or software and produces a score in the range [0, 10]. It quantifies the vulnerability  risk score ($L$) of $i^{th}$ VM in the range [0, 1] by applying Eq. (\ref{eq:vmvul}). 
\begin{equation} \label{eq:vmvul}
L_i = \frac{{V}^{Score}_i}{10} \quad \forall i \in [1,Q],  {V}^{Score} \in [1,10]
\end{equation}

\subsection{Hypervisor vulnerability}
The security risk of a hypervisor (${H}$) depends on its own vulnerability (${S}^{Hyp\_scor}$) as computed in Eq. (\ref{hypvul 1}) by applying CVSS score system and the vulnerability of the VMs hosted on it. The overall vulnerability score of hypervisor ${H}_j$ is given by Eq. (\ref{hypvul 2}), where  $max({L}_i \times \omega_{ij})$ represents maximum VM vulnerability score ($L$) among all VMs hosted on server $S_j$,  $\forall i \in [1,Q], j \in [1,P]$, $\omega_{ij}=1$ if ${S}_j$ hosts ${V}_i$. 

\begin{gather} 
{S}_{j}^{Hyp\_scor} = \frac{{S}^{Score}_j}{10} \quad \forall {S}^{Score} \in [1, 10]  \label{hypvul 1}\\
\int_{t_a}^{t_b}{H}_jdt = \int_{t_a}^{t_b}{S}^{Hyp\_scor}_j(1+max(L_i \times \omega_{ij}) ) dt \label{hypvul 2} 
\end{gather}

\section{ Allocation-dependent factors} \label{allocation}
{The cybersecurity risk factors pertaining to the distribution of physical resources and assignment of VMs on physical servers subject to resource availability constraints, characterize allocation-dependent risk factors.} The VM security risks due to the {Side-channel effect} and {Network cascading effect} depend upon the placement of VMs of different users on available servers (i.e., $\mathds{U} \times \mathds{V}\Rightarrow\mathds{S}$).  Two VMs $V_i$ and $V_j$ are inter-dependent $iff (V_i, V_j) \in \mathds{LA} $, where $\mathds{LA}$ implies {Legal Access} subject to the constraints: 
\begin{itemize}
	\item  $V_i(\mathds{LA}) V_i \quad  \forall_{V_i}\in\mathds{LA} $,
	\item  $V_i(\mathds{LA}) V_j =V_j(\mathds{LA}) V_j \quad \forall_{ V_i, V_j} \in \mathds{LA}$,
	\item  $\{V_i(\mathds{LA}) V_j \cup V_j(\mathds{LA}) V_k\} \Rightarrow V_i(\mathds{LA}) V_k$,
	\item $ \forall_{ V_i, V_j, V_k} \in \mathds{LA}$ 
\end{itemize}
%(\textit{i}) $V_i(\mathds{LA}) V_i \quad  \forall_{V_i}\in\mathds{LA} $, (\textit{ii}) $V_i(\mathds{LA}) V_j =V_j(\mathds{LA}) V_j \quad \forall_{ V_i, V_j} \in \mathds{LA}$, (\textit{iii}) $ \{V_i(\mathds{LA}) V_j \cup V_j(\mathds{LA}) V_k\} \Rightarrow V_i(\mathds{LA}) V_k\quad \forall_{ V_i, V_j, V_k} \in \mathds{LA}$. 
As depicted in Fig. \ref{fig:coresidencynetwork-attack}, a malicious user $U^{Mal}$ executes an application at $V_1$ hosted on the server ($S_1$) having an effective VM vulnerability, i.e., $L_1^1 > L_{Thr}$, achieves co-residency with one of the inter-dependent VMs ($\{V_1, V_2, ..., V_Z\} \in \mathds{LA}$), where $Z$ is the number of inter-dependent VMs. The malicious VM ($V^{Mal}_1$) successfully launches {side-channel threat} on vulnerable VM ($V_2$) and the threat propagates to multiple VMs crossing physical boundaries of network devices using {network cascading effect} via inter-communication links among VMs: $\{V_1, V_2, ..., V_Z \} \in \mathds{LA}$. The probability of threat ($\mathds{P}(\Xi_i)$) for $i^{th}$ VM over time-interval \{$t_a, t_b$\} is defined using Eq. (\ref{eq_Network_threat}),
where $\mathds{C}^{\ast}_{i{i^\ast}j}=\omega_{ij}\times \omega_{{i^\ast}j}=\{0, 1\}$ is a Boolean variable which specifies co-location between $i^{th}$ VM ($V_i$) and ${i^\ast}$ malicious VM ($V_i^\ast$) at server ($S_j$).  
%$\omega_{ij}\times \omega_{kj} \times L_k$
\begin{equation}
\label{eq_Network_threat}
\mathds{P}(\Xi_i)= \begin{cases}
1, & {If( (L_i > L_{Thr} \lor {H}_j > {H}_{Thr}) \land \mathds{C}^{\ast}_{i{i^\ast}j}}) \\
1, & {If(\Pi_{k=1}^Z(\mathds{C}^{\ast}_{k{i^\ast}{j^\ast}}\times \mathds{C}^{\ast}_{i{k}j} \times L_k)>L_{Thr})} \\
0, & {\text{otherwise}}  
\end{cases}	
\end{equation}
\begin{figure}[!htbp]
	\centering
	\includegraphics[width=0.63\linewidth]{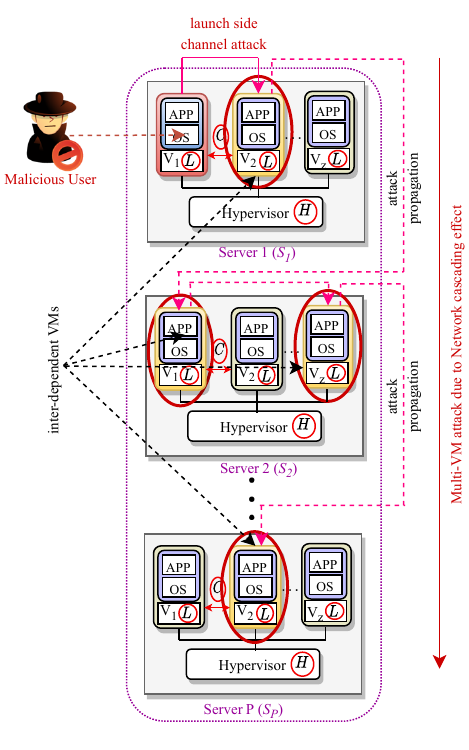}
	\caption{Side channel and Network cascading threats}
	\label{fig:coresidencynetwork-attack}
\end{figure}

\subsection{Side-channel effect}
Let a malicious VM $V_j^{Mal}$ and benign VM $V_i^{Ben}$ are hosted on server ${S}_k$. If $V_j^{Mal}$ compromises any VM on server ${S}_k$, then it can compromise other co-resident VMs eventually. Hence, the survival of $V_i^{Ben}$ depends on its own vulnerability score ($L_i$) and its co-resident VMs vulnerability score. The side-channel risk score (${C}$) of $V_i^{Ben}$ during time-interval $\{t_a, t_b\}$ is calculated as stated in Eq. (\ref{corres}), where  $\omega_{jk} \times \omega_{ik}$ represents co-location of $i^{th}$ and $j^{th}$ VM on $k^{th}$ server, $\forall i, j \in [1, Q]$, $k \in [1, P]$.
\begin{equation}
\int_{t_a}^{t_b}{C}_idt = \int_{t_a}^{t_b}1-\Pi_{j=1}^Q(1-L_j \times \omega_{jk} \times \omega_{ik}) dt \label{corres}
\end{equation}	

\subsection{Network cascading effect } The impact of cascading network connections among VMs on cloud  security establishes the network cascading effect. It is computed 
in terms of network cascading score (${N}$) respective to VM $V_i$ during time-interval $\{t_a, t_b\}$ using Eq. (\ref{networkrisk}), where ${V}_i$ and ${V}_j$ are connected via legal access network link and hosted on different servers $S_k$ and $S_{k^\ast}$ such that $\forall i, j \in [1,Q], i \ne j$. If a malicious VM ${V}^{Mal}$ hosted on server $S_{k^\ast}$, is successful in compromising the VM $V_j$, then  it can compromise VM ${V}_i$ and all other VMs that are connected via common network by exploiting the network paths.  
\begin{equation}
	\int_{t_a}^{t_b}{N}_i dt = \int_{t_a}^{t_b}1-\Pi_{j=1}^Q(1-L_j\times \omega_{ik} \times \omega_{jk^\ast}) dt 	\label{networkrisk}
	\end{equation}

%Algorithm \ref{algo-otp-mlb} entails the operational summary of OTP-MLB Framework based secure and energy-efficient VM placement.
%\section{Operational Flow and Complexity analysis }
\section{Operational Design and Complexity}
MR-TPM utilizes values of current state of multiple security attack factors \{${R}_1, {R}_2, {R}_3, {R}_4, {R}_5$\} and three historical databases namely $(i)$ VMs' resource utilization \{$RU_1$, $RU_2$, ..., $RU_Q$\} $\in \mathds{RU}_{db}$, $(ii)$ user-records \{$U_1$, $U_2$, ..., $U_M$\} $\in \mathds{U}_{db}$ and $(iii)$
VM threats traces \{$\Xi_1$, $\Xi_2$, ..., $\Xi_n$\} $\in \mathds{Th}_{db}$. The set of users $U_1$, $U_2$, ..., $U_M$, servers $S_1$, $S_2$, ..., $S_P$ and $V_1$, $V_2$, ..., $V_Q$ are initialized followed by a mapping $\mathds{U} \times \mathds{V}\Rightarrow\mathds{S}$ among VMs, users and servers. The VMs are allocated to servers using some suitable VM placement strategy, for example, First-Fit Decreasing (FFD), Best-Fit, Greedy, Random-Fit etc. Thereafter, for each consecutive time-intervals \{$t_a$, $t_b$\}, current resource utilization of $V_1$, $V_2$, ..., $V_Q$ are passed as input into a workload predictor \cite{saxena2020proactive} trained with $\mathds{RU}_{db}$ to estimate their resource utilization during next time-interval. The threat status prediction is conducted for the VMs with predicted workload estimation ($W^p>0$). To predict future threat status of VM $V_i$, values of ${R}_1, {R}_2, {R}_3, {R}_4$ associated with $V_i$ are assessed by applying Eqs. (\ref{eq:vmvul})-(\ref{networkrisk}). The assessment of ${R}_5$ is done by analysing the behaviour of co-resident users of $V_i$ by using RFC based user classifier trained with $\mathds{U}_{db}$. The current score values of ${R}_1, {R}_2, {R}_3, {R}_4, {R}_5$ are fed as input into threat predictor ($\mathds{TP}^{\ast\ast}$) trained and tested with $\mathds{Th}_{db}$, to predict the future threat status ($\hat{V}^\Xi_i$) of $V_i$.
 Accordingly, the VMs with $\hat{V}^\Xi_i>0$ are migrated to server where $\hat{V}^\Xi_i=0$ by applying Eq. (\ref{mig}). The migration cost is computed using Eq. (\ref{mig_cost}), where $\mathds{D}(S_k ,S_j)$ is the distance or number of hops covered by migrating VM $V_{mig}$ from source ($S_k$) to destination server $S_j$, \{$j, k \in [1,P]$\}, $V_{mig} \in \mathds{TP}_{V}$, $\mathds{WW}(V_{mig})$ = $V_{mig}^{{CPU}} \times V_{mig}^{{Mem}}$ is the size of migrating VM, $\mathds{TP}_{V}$ is the list of VMs with `unsafe' status or VMs on overloaded server ($S_k$).  The first term $\sum{\mathds{CC}_{mig.j}{\mathds{D}(S_k ,S_j) *\mathds{WW}(V_{mig})}}$ signifies network energy consumed during VM migration. The second term $\sum{\mathds{G}_{j}{\mathds{B}_{j}}}$ specifies server state transition energy, where if $i^{th}$ VM is placed at $j^{th}$ server after migration, then $\mathds{CC}_{mig.j}=1$, otherwise $0$. If $j^{th}$ server receives one or more VMs after migration, then $\mathds{G}_{j}=1$ else it is 0. Similarly, if $\mathds{B}_{j}=0$, then $j^{th}$ server is  active before migration, otherwise, $\mathds{B}_{j}=\mathds{E}_{tr}$ where $\mathds{E}_{tr} = 4260 $ Joules which is energy consumed in switching a server from sleep to active state \cite{saxena2022fault, saxena2022sre}. 
 \begin{gather}\label{mig}
 {{V}_i}^{mig\_status}=\begin{cases}
 1  & {If(\hat{V}^\Xi_i>0)} \\
 0  & {\text{otherwise}} 
 \end{cases}
 \\ \label{mig_cost}
 \resizebox{0.48\textwidth}{!}{$  
 	\mathchorus{M_{c}}={\sum{\mathds{CC}_{mig.j}(\mathds{D}(S_k ,S_j)*\mathds{WW}(V_{mig}))} + \sum{\mathds{G}_{j}{\mathds{B}_{j}}}}
 	$}
 \end{gather}
The operational summary for proposed work is depicted in Algorithm \ref{algo-otp-mlb}.  
\begin{figure}[!htbp]
	\removelatexerror
\begin{algorithm}[H]
\caption{MR-TPM Operational Summary}
\label{algo-otp-mlb}
Initialize: $List_{\mathds{U}}$, $List_{\mathds{V}}$, $List_{\mathds{S}}$,  $\omega$\; 
Allocate $V_1$, $V_2$, ..., $V_Q$ to $S_1$, $S_2$, ..., $S_P$ by defining a mapping $\mathds{U} \times \mathds{V}\Rightarrow\mathds{S}$   \;		
\For {each time-interval $\{t_a, t_b\}$}{ 
     [$V^{Pred}_i$] $\Leftarrow$ Workload Prediction($V_i$) $\forall i \in \{1, 2, ...Q\}$ \;
      \eIf {($V^{Pred}_i>0$)}{
                    [$\hat{V}^\Xi_i>0$] $\Leftarrow$ Threat Predictor ($\mathds{TP}^{\ast\ast}$)\;
                   	\eIf{ $\hat{V}^\Xi_i ==$ `unsafe'}{                   	
                   	Migrate ${V}_i$ to server ${S}_k$ such that $\hat{V}^\Xi_i ==$ `safe'\;
                    Compute $\mathchorus{M_{c}}$ by applying Eq. (\ref{mig_cost})\;
                }
               {Keep ${V}_i$ at same server until user terminates it\;
             }          
 } 
{ VM threat prediction is not required\;}
}
\end{algorithm}
\end{figure}
 Step 1 initializes the list of VMs, servers, users (owners of these VMs)  producing a complexity of $\mathcal{O}(1)$. The time complexity of step 2 depends on the type of chosen VM placement policy. Steps 3-16 repeat for $Y$ time intervals, wherein the step 4 has complexity of $ \mathcal{O}(W)$ \cite{saxena2020proactive}. The complexity of step 6 is $T\Leftarrow \mathcal{O}(thzlogn)$, where $t$ is the number of trees, $h$ is the height of the trees, and $z$ is the number of non-missing entries in the training data. Prediction for a new sample consumes time $\mathcal{O}(th)$. Therefore, the total time-complexity of MR-TPM operational algorithm is $\mathcal{O}(YWT)$.

\section{Performance Evaluation }
\subsection{Experimental Set-up and Implementation}
The simulation experiments are executed on a server machine assembled with two Intel\textsuperscript{\textregistered} Xeon\textsuperscript{\textregistered} Silver 4114 CPU with 40 core processor and 2.20 GHz clock speed in Cloud data center simulation framework implemented in Python Jupyter Notebook. The computation machine is deployed with 64-bit Ubuntu 16.04 LTS, having main memory of 128 GB. The datacenter environment is set up with three different types of servers and four types of VMs configuration shown in Tables \ref{table:server} and \ref{table:vm}. The resource features like power consumption ($P_{max}, P_{min}$), MIPS, RAM and memory are taken from real server configuration; IBM \cite{IBM1999} and Dell \cite{Dell1999},  where $S_1$ is `ProLiantM110G5XEON3075', $S_2$ is `IBMX3250Xeonx3480' and $S_3$ is `IBM3550Xeonx5675'. Furthermore, the experimental VMs configurations are inspired from the VM instances of the Amazon website \cite{amazon1999EC2}. 

\begin{table}[htbp]
	\centering
	
	\caption[Table caption text] {Server Configuration}  %\cite[p.10]{refid} }
	\label{table:server}
	%\resizebox{0.8\textwidth}{!}{\begin{minipage}{\textwidth}
	\resizebox{9cm}{!}{
		\begin{tabular}{lccccccc}
			\hline
			%\multicolumn{2}{c}{Item} \\
			%\cline{1-2}
			Server&PE&MIPS&RAM(GB)&Storage(GB)&$PW_{max}$&$PW_{min}$/$PW_{idle}$\\
			\hline
			$S_1$ 	& 2&2660&4&160&135&93.7 \\
			$S_2$	& 4&3067&8&250&113&42.3 \\
			$S_3$	& 12&3067&16&500&222&58.4 \\

			\hline
	\end{tabular}}
	%	\end{minipage}}
\end{table}

\begin{table}[htbp]
	\centering
	
	\caption[Table caption text] {VM Configuration}  %\cite[p.10]{refid} }
	\label{table:vm}
	%	\resizebox{0.8\textwidth}{!}{\begin{minipage}{\textwidth}
	\begin{tabular}{lcccc}
		\hline
		%\multicolumn{2}{c}{Item} \\
		%\cline{1-2}
		VM type& PE &MIPS&RAM(GB)&Storage(GB)\\
		\hline
		$v_{type1}$&1&500&0.5&40\\
		$v_{type2}$&2&1000&1&60\\
		$v_{type3}$&3&1500&2&80\\
		$v_{type4}$&4&2000&3&100\\

		\hline
	\end{tabular}
	%	\end{minipage}}
\end{table}
\subsection{Datasets and Simulation parameters}
MR-TPM is evaluated using two benchmark VM  traces from a publicly available real workload datasets: \textit{OpenNebula Virtual Machine Profiling Dataset} (ONeb) \cite{24mb-vt61-20} and \textit{Google Cluster Data} (GCD) \cite{reiss2011google}.   ONeb  provides  information  gathered by the monitoring system for six VMs  over 63 Hours via executing a set of probe programs provided by OpenNebula. It reports VM threats respective to  the server status, basic performance indicators, as well as VM status, and resource capacity consumption of server hosting these VMs. The exact values of VM and hypervisor vulnerability scores are not reported in the original VM threat database. Therefore, to prepare VM threat database including attributes: \{${V\_id}_i^{victim}$, ${S}\_id$, ${V\_id}^{Mal}$, ${V}_i^{CPU}$, ${V}_i^{BW}$, ${V}_i^{memory}$, ${R}^{score}_i$, $L_i$, ${H}_i$, ${C}_i$, ${N}_i$, ${V}_i^{status}$, ..., etc.\}, the VMs of ONeb dataset that have experienced attacks, are assigned vulnerability score higher than the threshold value of VM threat (which is considered 0.5 for the experiments) and the rest of the risk scores are computed using Eqs. (\ref{eq_VM_Hypr_threat})-(\ref{networkrisk}). These VM threats information is learned by the VM threat predictor for estimation of threats  on VMs before occurrence.
\par Also, we have utilized a realistic workload of Google Cluster Data (GCD)  \footnote{https://github.com/HiPro-IT/CPU-and-Memory-resource-usage-from-Google-Cluster-Data}, which provides resource usage percentage for each job in every five minutes over period of twenty-four hours. GCD contains capacity usage information of resources CPU, memory, disk I/O  request information of 672,300 jobs executed on 12,500 servers for the period of 29 days. The VM vulnerability ($L$) and server hypervisor vulnerability (${H}$) are generated in the range [0, 10] during VMs and the server's initialization.  Accordingly, the VM threat database reporting traces of attacks on GCD VMs, including attributes \{${V\_id}_i^{victim}$, ${S}\_id$, ${V\_id}^{Mal}$, ${V}_i^{CPU}$, ${V}_i^{BW}$, ${V}_i^{memory}$, ${R}^{score}_i$, $L_i$, ${H}_i$, ${C}_i$, ${N}_i$, ${V}_i^{status}$, ..., etc.\} is generated and updated at runtime according to requirement of the proposed model.
These datasets do not report user database and hence, we created a user database consisting of \{$U_{id}$, $Attack_{threshold}$, $U_{class}$\} and utilized it for user behaviour analysis based on their previous VM usage. The number of users is considered equals to 30\% of the number of VMs (i.e., 1200 VMs), who requested varying number and type of VMs over time. Therefore, different number and types of VMs are mapped with user at run-time and according to the risk scores associated to different VMs, threat is defined. Each user can hold VMs with a constraint that at any instance, the total number of VMs requests must not exceed total number of available VMs at the datacenter. The user database is created and updated during runtime. All the experiments are executed for 100 time-intervals of five minutes to analyse the performance of proposed model dynamically, though this period can be extended as per the availability of traces. The security threats are generated depending upon the threshold values for the four types of risks \{$L_i$, ${H}_i$, ${C}_i$, ${N}_i$\} associated with a VM and presence of the malicious ${V}^{Mal}$. The presence of some malicious user VM (${V}^{Mal}$) on a server and the risk scores corresponding to $i^{th}$ VM (${V}_i$) `greater than equal to' their respective threshold values indicate the high probability of security threat (i.e., $V_i^{\Xi} > 0$). 
\subsection{VM Cyberthreat Prediction} \label{results_prediction}
The VM threat prediction is performed for the different population of malicious user $U^{Mal}$, such as 5\%, 20\%, 50\% and 90\%. The prediction accuracy achieved during training, testing, and live phase for 5\% and 50\% $U^{Mal}$  over period of 500 minutes is shown in Fig. \ref{fig:accuracy} for both the datasets. It can be noticed  that prediction accuracy is closer to 98\% for all three phases, which is slightly increasing for live cyberthreat detection during each experiment because of the capability of online-learning and re-training of MR-TPM with the passage of time. To provide online training, we perform read/write operation of live VM threats in threat database file dynamically during period of 500 minutes. The Receiver Operator Characteristics (ROC)
curve obtained for different experiments using both the datasets are depicted in Fig. \ref{fig:roc}.   
 ROC curves for the $U^{Mal}=5\%$ is better than the ROC curves obtained with $U^{Mal}=50\%$ because of effective learning of true threats  in the presence of least number of malicious users. It is
observed that the proposed MR-TPM efficiently predicts VM threats for the test as well as live data in all the experiments for both datasets. 
\par Fig. \ref{fig:threatprediction} analyses the Actual Threat (AT), Predicted Threat (PT), and Unpredicted Threat (UT) for online VM threat prediction  in the presence of 5\%, and 50\% ${U}^{Mal}$ for both datasets. It can be observed that most of the VM threats are predicted correctly where UT is closer to zero and PT is closer to AT, indicating that along with all true threats, some false threats are also predicted. However, the difference between AT and PT is reducing over the time with enhancement of learning capability of VM threat predictor. 
\begin{figure*}[htbp]
	
	\centering
	\subfigure[GCD with ${U}^{Mal}$ = 5\% ]{\includegraphics[width=.24\textwidth]{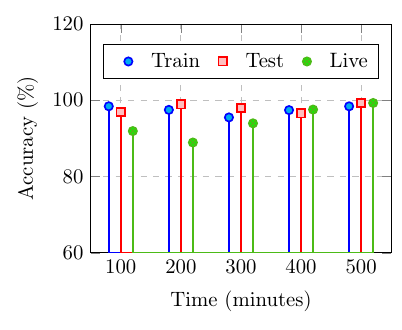}}\hfill
	\subfigure[GCD with ${U}^{Mal}$ = 50\% ]{\includegraphics[width=.24\textwidth]{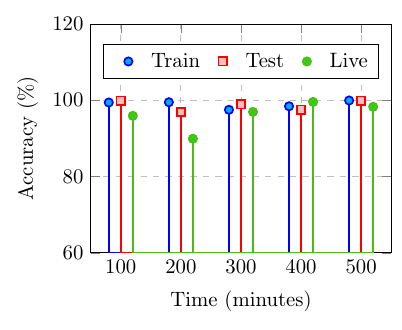}}\hfill
	\subfigure[ONeb with ${U}^{Mal}$ = 5\% ]{\includegraphics[width=.24\textwidth]{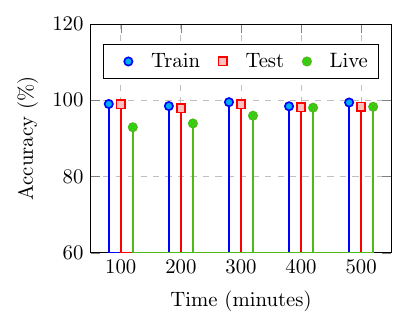}}
	\subfigure[ONeb with ${U}^{Mal}$ = 50\% ]{\includegraphics[width=.24\textwidth]{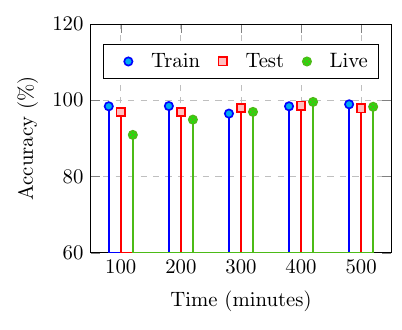}}
	
	\caption{Prediction Accuracy }
	\label{fig:accuracy}
	
\end{figure*} 
\begin{figure*}[htbp]
	
	\centering
	\subfigure[GCD with ${U}^{Mal}$ = 5\% ]{\includegraphics[width=.24\textwidth]{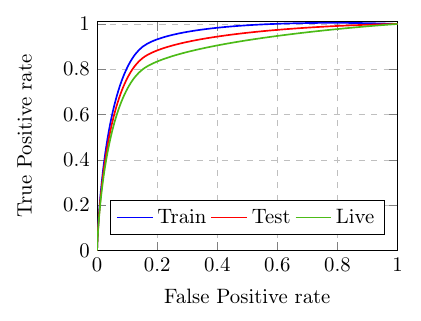}}\hfill
	\subfigure[GCD with ${U}^{Mal}$ = 20\% ]{\includegraphics[width=.24\textwidth]{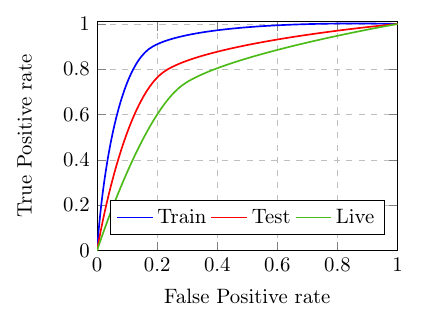}}\hfill
	\subfigure[ONeb with ${U}^{Mal}$ = 50\% ]{\includegraphics[width=.24\textwidth]{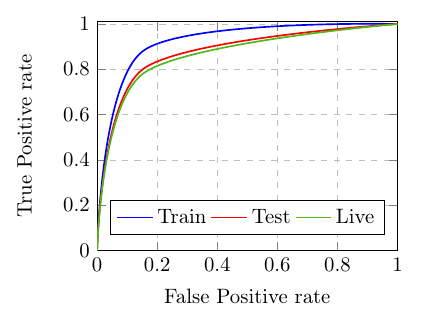}}
	\subfigure[ONeb with ${U}^{Mal}$ = 90\% ]{\includegraphics[width=.24\textwidth]{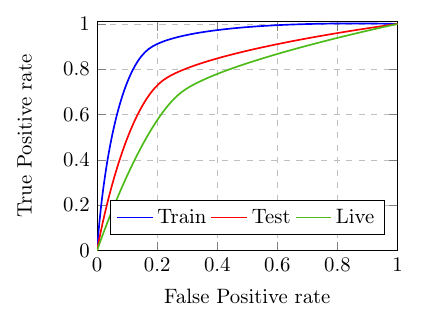}}
	
	\caption{ROC Curves}
	\label{fig:roc}
	
\end{figure*} 
\begin{figure*}[htbp]
	
	\centering
	\subfigure[GCD with ${U}^{Mal}$ = 5\% ]{\includegraphics[width=.24\textwidth]{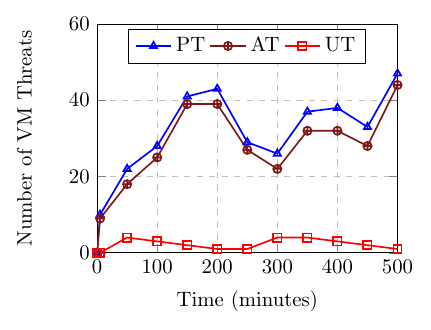}}\hfill
	\subfigure[GCD with ${U}^{Mal}$ = 50\% ]{\includegraphics[width=.24\textwidth]{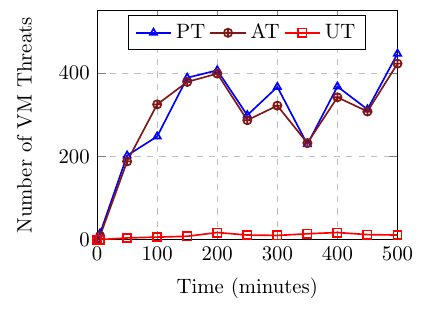}}
		\subfigure[ONeb with ${U}^{Mal}$ = 5\% ]{\includegraphics[width=.24\textwidth]{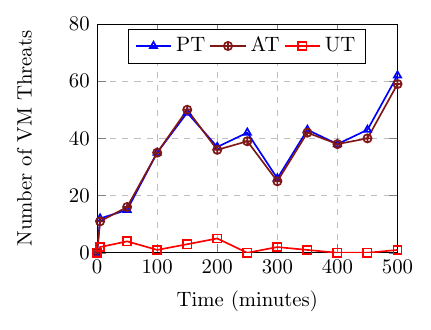}}\hfill
	\subfigure[ONeb with ${U}^{Mal}$ = 50\% ]{\includegraphics[width=.24\textwidth]{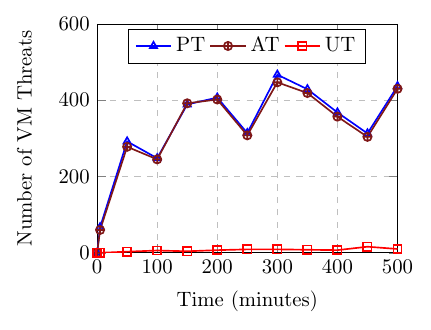}}	
	\caption{Number of threats (AT:Actual Threats, PT: Predicted Threats, UT: Unpredicted Threats) }
	\label{fig:threatprediction}
	
\end{figure*} 
\par  The values of precision, recall, F1 measure score, average of mean square error ($Avg. MSE$), average of mean absolute error ($Avg. MAE$) observed for the different experimental cases of both the datsets, including GCD and ONeb VM traces are shown in Table \ref{table:performance} and Table \ref{table:performance_nebula} which are consistently higher than $0.95$ for each case. The $Avg. MSE$ and $Avg. MAE$ values are observed in the range of [$ 0.0001-0.0008$] and [$ 0.001-0.010$], respectively, and the accuracy of prediction is higher than 96\% reaching up to 99.71\% and 99.25\% for the GCD and ONeb, respectively. The reason for such an accurate prediction is the incremental learning of MR-TPM with historical and live VM threat databases in real-time. Fig. \ref{fig:riskanalysis} shows the changes observed in the various risks scores \{$L$, $H$, $N$, $C$\} of a randomly selected VM among the 1200 VMs under simulation for a period of 500 minutes.
\begin{table}[!htbp]
	
	\caption[Table caption text] {Performance metrics for GCD VM traces}  
	\label{table:performance}
	\small
	%\centering 
	%\resizebox{0.8\textwidth}{!}{\begin{minipage}{\textwidth}
	\resizebox{9cm}{!}{
		\centering
		\begin{tabular}{|l|c|c|c|c|c|c|c|}
			\hline
			${U}^{Mal}$	&$time$ & \multicolumn{6}{c|}{Performance metrics} \\
			\cline{3-8} \cline{3-8}
			(\%)&$(min.)$& $Precision$ &$Recall$ &$F1 score$&${Avg. MSE}$&$Avg. MAE$&$Accuracy$ \\ \hline \hline			
			\multirow{5}{*}{5}&100&0.97&0.99&0.98&0.00021&0.0034 &98.87 \\ \cline{2-8}
			&200&0.99&0.99&1.00&0.00043&0.0294 &98.96\\ \cline{2-8}
			&300&1.00&0.99&0.96&0.00036&0.0052 &99.42\\ \cline{2-8}
			&400&0.99&0.97&1.00&0.00071&0.0073 &98.87 \\ \cline{2-8}
			&500&1.00&0.99&0.99&0.00191&0.0099 &99.71 \\ \hline \hline
			\multirow{5}{*}{20}&100&1.00&0.98&0.98&0.00062&0.0094 &97.07 \\ \cline{2-8}
			&200&0.99&0.99&0.99&0.00031&0.0044 &98.66\\ \cline{2-8}
			&300&1.00&1.00&1.00&0.00026&0.0025 &99.67\\ \cline{2-8}
			&400&0.99&0.97&1.00&0.00043&0.0069 &99.17 \\ \cline{2-8}
			&500&0.99&0.98&0.99&0.00061&0.0084 &99.96 \\ \hline \hline
			
			\multirow{5}{*}{50}&100&0.97&0.99&0.98&0.00039&0.0041 &98.16 \\ \cline{2-8}
			&200&0.98&0.96&0.99&0.00058&0.0064 &97.43\\ \cline{2-8}
			&300&1.00&1.00&1.00&0.00016&0.0015 &99.69\\ \cline{2-8}
			&400&0.99&0.98&0.98&0.00037&0.0049 &98.25 \\ \cline{2-8}
			&500&0.99&0.99&0.99&0.00031&0.0034 &98.17 \\ \hline \hline
			\multirow{5}{*}{90}& 100&1.00&0.99&1.00&0.00028&0.0041 &98.76 \\ \cline{2-8}
			&200&0.98&0.96&0.99&0.00058&0.0064 &99.74\\ \cline{2-8}
			&300&1.00&0.99&1.00&0.00014&0.0020 &98.69\\ \cline{2-8}
			&400&0.99&0.98&0.96&0.00029&0.0039 &99.25 \\ \cline{2-8}
			&500&1.00&0.99&0.99&0.00011&0.0014 &99.97 \\ \hline 
			
	\end{tabular}}
\end{table}
 \begin{table}[!htbp]
 	
 	\caption[Table caption text] {Performance metrics for OpenNebula VM traces}  
 	\label{table:performance_nebula}
 	\small
 	%\centering 
 	%\resizebox{0.8\textwidth}{!}{\begin{minipage}{\textwidth}
 	\resizebox{9cm}{!}{
 		\centering
 		\begin{tabular}{|l|c|c|c|c|c|c|c|}
 			\hline
 			${U}^{Mal}$	&$time$ & \multicolumn{6}{c|}{Performance metrics} \\
 			\cline{3-8} \cline{3-8}
 			(\%)&$(min.)$& $Precision$ &$Recall$ &$F1 score$&${Avg. MSE}$&$Avg. MAE$&$Accuracy$ \\ \hline \hline			
 			\multirow{5}{*}{5}&100&0.96&0.98&0.96&0.0007&0.0014 &99.10 \\ \cline{2-8}
 			&200&0.99&0.99&1.00&0.00023&0.0094 &99.06\\ \cline{2-8}
 			&300&.99&0.99&0.98&0.00045&0.0005 &99.10\\ \cline{2-8}
 			&400&1.00&0.97&0.99&0.00091&0.0023 &99.07 \\ \cline{2-8}
 			&500&0.99&0.98&1.00&0.00071&0.0006 &99.11 \\ \hline \hline
 			\multirow{5}{*}{20}&100&0.99&0.97&0.97&0.00062&0.0004 &98.16 \\ \cline{2-8}
 			&200&1.00&0.99&1.00&0.00051&0.0024 &97.96\\ \cline{2-8}
 			&300&1.00&1.00&0.99&0.00025&0.0015 &98.17\\ \cline{2-8}
 			&400&1.00&0.99&1.00&0.00022&0.0029 &99.17 \\ \cline{2-8}
 			&500&0.99&0.98&1.00&0.00021&0.0014 &99.01 \\ \hline \hline
 			
 			\multirow{5}{*}{50}&100&0.98&0.99&0.99&0.00031&0.0031 &98.82 \\ \cline{2-8}
 			&200&0.98&0.96&0.99&0.00058&0.0045 &98.03\\ \cline{2-8}
 			&300&1.00&1.00&1.00&0.00016&0.0021 &99.04\\ \cline{2-8}
 			&400&0.98&1.00&0.99&0.00027&0.0030 &99.25 \\ \cline{2-8}
 			&500&1.00&0.99&0.98&0.00040&0.0032 &98.85 \\ \hline \hline
 			\multirow{5}{*}{90}& 100&0.99&1.00&1.00&0.00108&0.0064 &97.91 \\ \cline{2-8}
 			&200&0.99&0.98&1.00&0.00078&0.0042 &98.23\\ \cline{2-8}
 			&300&0.99&0.97&0.99&0.00067&0.0038 &98.99\\ \cline{2-8}
 			&400&0.99&0.99&0.98&0.00029&0.0029 &99.06 \\ \cline{2-8}
 			&500&0.99&0.99&1.00&0.00016&0.0016 &99.17 \\ \hline 
 			
 	\end{tabular}}
 \end{table}
 
\begin{figure}[!htbp]
	\centering
	\includegraphics[width=0.6\linewidth]{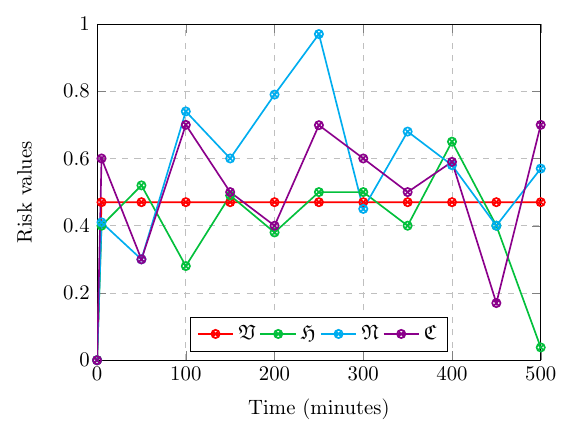}
	\caption{ Variation of multiple risk values for a VM}
	\label{fig:riskanalysis}
\end{figure}

\subsection{Deployment and Comparison}

To further analyse the efficiency of MR-TPM, it is deployed with existing state-of-the-art VM placement (VMP) policies, including Previously Selected Server First (PSSF) \cite{han2017using}, Secure and Energy Aware Load Balancing (SEA-LB) \cite{singh2019secure}, Security Embedded Dynamic Resource
Allocation (SEDRA) \cite{saxena2020security} and baseline VMP policies, including First-Fit Decreasing (FFD), Best-Fit (BF), and Random- Fit (RF). All the results shown in Section \ref{results_prediction} are derived with FFD VMP policy. 
%\par \textit{Conceptual analysis of selected pioneering works}: Han et al. proposed PSSF VM allocation policy in which every server maintains a list of cloud users that it has ever hosted at run time. For a user $U_i$ requesting a new VM, those servers which have previously hosted their VMs are considered first for assignment. In the absence of such server, the server having more capacity of remaining resources, is considered. In PCUF, previously co-located users are preferred to be allocated together on a server during VM allocation. It aims to minimize the probability of malicious VM co-location. SEA-LB allocates VMs considering maximum resource utilization and minimum power consumption and side-channel attacks by applying modified genetic algorithm approach. The security is provided by minimizing the number of shared servers at the cost of resource utilization. Recently, Saxena et al. have presented a security embedded resource allocation (SEDRA) model in which performance of network traffic and inter-VM links are considered to detect and mitigate VM threats by utilizing a random tree classifier.
\par Table \ref{table:attackcomparison} compares the average number of VM threats realised without and with MR-TPM (results are shown for the Live phase) integrated together with the above mentioned VMP policies. It can be observed that up to 88.5\%, 86.5\%, 86.2\%, 88.9\%, 88.5\% and 88.1\% threats are reduced with proposed approach over PSSF, SEA-LB, SEDRA, RF, BF and FFD, respectively, for $U^{Mal}\% = 90$ at $T (min)=500$.  
\begin{table*}[!htbp]
	
	\caption[Table caption text] {Comparison of number of threats without and with MR-TPM (Live phase) deployed with various VMP approaches }  
	\label{table:attackcomparison}
	\small%\centering 
	%\resizebox{0.8\textwidth}{!}{\begin{minipage}{\textwidth}
	
	\centering
	\resizebox{14cm}{!}{
	\begin{tabular}{|l|l||c|c||c|c||c|c||c|c||c|c||c|c|}
		\hline
		${U}^{Mal}$	&$T$ & \multicolumn{12}{c|}{Percentage of VM security threats ($\Xi$)} \\
		\cline{3-14} \cline{3-14}
		(\%)&$(min)$&\multicolumn{2}{c||}{PSSF \cite{han2017using}}
		&\multicolumn{2}{c||}{SEA-LB \cite{singh2019secure} } 
		&\multicolumn{2}{c||}{SEDRA  \cite{saxena2020security}}
		&\multicolumn{2}{c||}{RF}
		&\multicolumn{2}{c|}{BF}
		&\multicolumn{2}{c|}{FFD} \\ \cline{3-14}
		& &W-$\mathds{TP}$&$\mathds{TP}$&W-$\mathds{TP}$&$\mathds{TP}$&W-$\mathds{TP}$&$\mathds{TP}$&W-$\mathds{TP}$&$\mathds{TP}$&W-$\mathds{TP}$&$\mathds{TP}$&W-$\mathds{TP}$&$\mathds{TP}$\\ \hline \hline			
		\multirow{5}{*}{5}&100&116& 19 &206& 13&137&17 &276&16 & 283& 19 &318&58 \\ \cline{2-14}
		&200&203&33&187&27&169&19&296 &22&258&17&287&22 \\ \cline{2-14}
		&300&270&27&193&19&125&18&226 &17&238&16&308&26\\ \cline{2-14}
		&400&216&32&208&18&148&23&298&17&222&18&256 &23\\ \cline{2-14}
		&500&223&22&214&21&177&25&196 &18&236&17&312&29\\ \hline \hline
		\multirow{5}{*}{20}& 100 &365&26& 327& 54&244 &17& 474& 19&86.7&17.4&678&35\\ \cline{2-14}
		&200&376&23&364&68&314&24&494 &17&89.9&15.7&673&42 \\ \cline{2-14}
		&300&399&17&397&39&344&28&478&29&87.5&14.0&579 &44\\ \cline{2-14}
		&400&416&31&389&28&297&28&473 &26&86.4&17.2&598 &58\\ \cline{2-14}
		&500&402&26&428&49&308&37&501 &28&89.7 &16.9&657&39\\ \hline \hline
		
		\multirow{5}{*}{50}& 100 &537& 28&466&37& 376&24& 638& 27&584&27&779&56\\ \cline{2-14}
		&200&556&23&451&39&349&23&627 &29&595&24&767&49\\ \cline{2-14}
		&300&536&15&485&44&339&30&692 &27&603&38&797 &67\\ \cline{2-14}
		&400&547&37&509&35&388&26&701&28&586 &27&745&54\\ \cline{2-14}
		&500&533&41&487&46&373&34&694 &25&567&15&779&48\\
		\hline \hline
		\multirow{5}{*}{90}& 100 &783&57&766& 41&676 &56&893&65&837&76&958&108\\ \cline{2-14}
		&200&723&49&748&56&621&43&907 &84&878&89&997&99\\ \cline{2-14}
		&300&792&78&687&75&658 &49&958&98&889&94&946 &83\\ \cline{2-14}
		&400&712&62&673&69&684 &57&897&87&847&88&984&94\\ \cline{2-14}
		&500&728&84&678&91&633 &87&927&102&837&96&996&118\\ \hline 
		
	\end{tabular}}
\end{table*}
The resource utilization of datacenter can be obtained using Eqs. (\ref{ru1}), (\ref{ru2}), where $Z$ is the number of resources, $\omega_{ji}=\{0, 1\}$ is mapping between server ($S_i$) and VM ($V_j$). Though in formulation, only ${CPU}$ and ${Mem}$ are considered, it is extendable to any number of resources.
\begin{equation}
{RU}_{dc}= \int\limits_{\substack{t_1\\\mathcal{}}}^{t_2} (\frac{	{RU}_{dc}^{{CPU}} +  {RU}_{dc}^{{Mem}} }{|Z|\times \sum_{i=1}^{P}{\gamma_i}}) dt\label{ru1}
\end{equation}
\begin{equation}
{RU}_{dc}^{r}=\sum_{i=1}^{P}{\frac{\sum_{j=1}^{Q}{\omega_{ji} \times V_j^{r}}}{S_i^{r}}} \quad r \in {{CPU}, {Mem}, etc.} \label{ru2}
\end{equation}
 The resource utilization of different VMP integrated with MR-TPM follows the trend: $FFD \ge SEA-LB \ge SEDRA \ge PSSF \ge BF \ge RF$, as shown in Fig. \ref{fig:pw}a.

The power consumption for $i^{th}$ server can be formulated as $PW_i$ and total power consumption $PW_{dc}$ during time-interval \{$t_1$, $t_2$\} can be computed by applying Eq. (\ref{power2}), where ${PW_i}^{max}$, ${PW_i}^{min}$ and ${PW_i}^{idle}$ are maximum, minimum and idle state power consumption of $i^{th}$ server.
\begin{equation}
PW_{dc} = 
\int\limits_{\substack{t_1\\\mathcal{}}}^{t_2} \sum_{i=1}^{P} {[{PW_i}^{max} - {PW_i}^{min}]{RU} + {PW_i}^{idle}} dt
\label{power2}
\end{equation}	
 Fig. \ref{fig:pw}b shows the comparison of power consumption which is highest (i.e., 109.10 KW) for MR-TPM + PSSF and least (69.29 KW) for MR-TPM + FFD. The average number of active servers is compared in Fig. \ref{fig:pw}c , where MR-TPM + FFD and MR-TPM + PSSF operates at the lowest (118) and highest (774) number of active servers, respectively. Both the power consumption as well as the number of active servers follow the trend: $FFD < BF < SEDRA < SEA-LB <RF< PSSF$. The reason for the resultant trend is that VMs are tightly packed onto servers using FFD,  while in the case of others, sharing of servers is minimized for the purpose of security at the cost of the larger number of active servers.

\begin{figure*}[htbp]
	
	\centering
	\subfigure[Resource utilization]{\includegraphics[width=.31\textwidth]{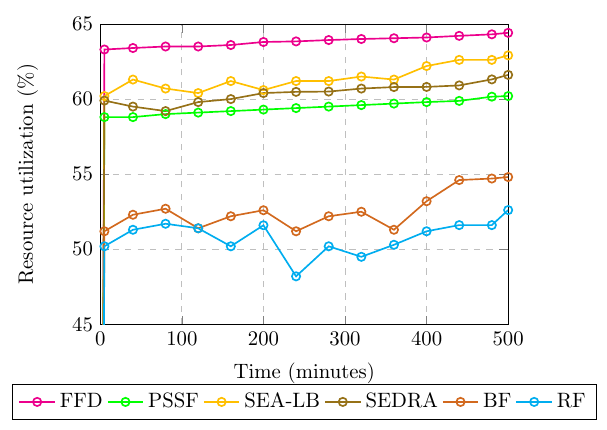}}\hfill
	\subfigure[Power consumption ]{\includegraphics[width=.31\textwidth]{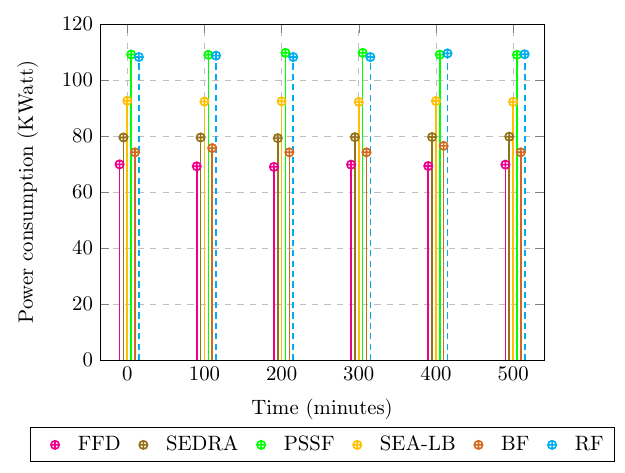}}\hfill
	\subfigure[Average number of active servers ]{\includegraphics[width=.31\textwidth]{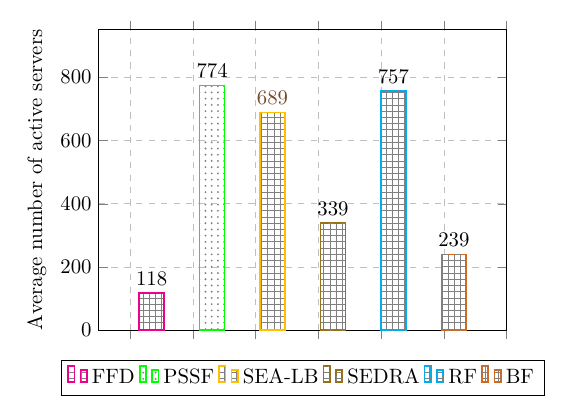}}
	%\subfigure[With ${U}^{Mal}$ = 90\% ]{\includegraphics[width=.24\textwidth]{ROC90}}
	
	\caption{Load management metrics}
	\label{fig:pw}
	
\end{figure*} 

\section{Conclusions and Future Work}
	To provide a comprehensive solution for secure workload distribution at cloud datacenter, a novel MR-TPM is proposed which estimates the future threats on user VMs by analysing multiple risk pathways, including VM and hypervisor vulnerabilities, co-residency, network cascading effects and user behaviour. The proposed model is periodically trained and retrained with historical and live threat data for accurate prediction of threat on VMs. MRTPM deployed with existing VM allocation policies minimizes  multiple risks based VM threats and related adversary breaches. The  performance evaluation of the proposed VM threat prediction model supports its efficacy in improving cybersecurity and resource efficiency over the compared approaches. In the future, MR-TPM can be extended with transfer learning to enhance its capabilities of analysing unknown/unseen security threats. Additionally, other possible security risk factors can be quantified and included to  improve the prediction approach of cyberthreats further.  
\section*{Acknowledgments}
The authors would like to thank the University of Aizu, Japan and the National Institute of Technology, Kurukshetra, India for financially supporting the research work.

\ifCLASSOPTIONcaptionsoff
\newpage
\fi

%\bibliographystyle{IEEEtran}
% argument is your BibTeX string definitions and bibliography database(s)
%\bibliography{IEEEabrv,../bib/paper}
%\bibliography{mybibfile}
% biography section
% 
% If you have an EPS/PDF photo (graphicx package needed) extra braces are
% needed around the contents of the optional argument to biography to prevent
% the LaTeX parser from getting confused when it sees the complicated
% \includegraphics command within an optional argument. (You could create
% your own custom macro containing the \includegraphics command to make things
% simpler here.)
%\begin{IEEEbiography}[{\includegraphics[width=1in,height=1.25in,clip,keepaspectratio]{mshell}}]{Michael Shell}
% or if you just want to reserve a space for a photo:

% if you will not have a photo at all:
\bibliographystyle{IEEEtran}
% argument is your BibTeX string definitions and bibliography database(s)
%\bibliography{IEEEabrv,../bib/paper}
\bibliography{mybibfile}
\begin{IEEEbiography}[{\includegraphics[width=0.7\linewidth]{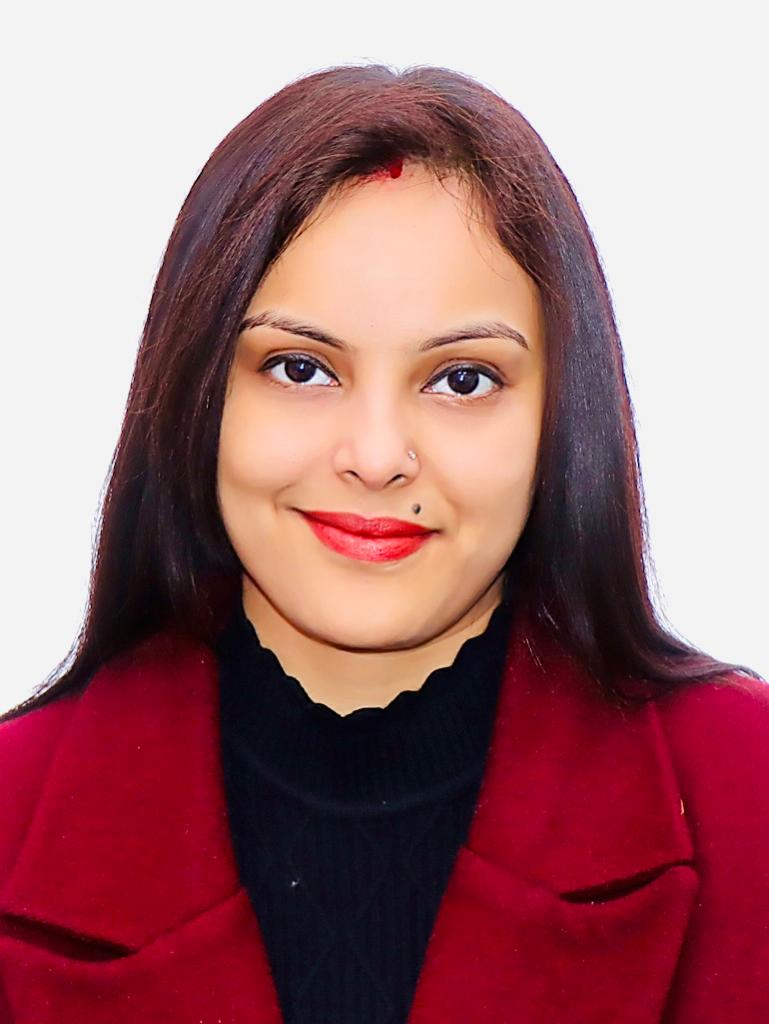}}]{Deepika Saxena}
	holds the position of an Associate Professor in the Division of Information Systems at the University of Aizu, Japan. Also, she is working as Part-time/Visiting Professor in the University of Economics and Human Sciences, Warsaw, Poland, Europe. She earned her Ph.D. degree in Computer Science from the National Institute of Technology, Kurukshetra, India, and completed her Post Doctorate from the Department of Computer Science at Goethe University, Frankfurt, Germany. She has been honored with the prestigious EUROSIM 2023 Best Ph.D. Thesis Award.  Her major research interests include Neural networks, Evolutionary algorithms,  Resource management, and Security in Cloud Computing, Internet traffic management, and Quantum machine learning, DataLakes, Dynamic Caching Management. Some of her research findings are published in top cited journals such as IEEE TSC, IEEE TSMC, IEEE TPDS,  IEEE TCC, IEEE Communications Letters, IEEE Networking Letters, IEEE Systems Journal,  IEEE Wireless Communication Letters, IEEE TNSM, IET Electronics Letters, Neurocomputing, etc.
\end{IEEEbiography}
\vskip 0pt plus -1fil 
\begin{IEEEbiography}[{\includegraphics[width=0.7\linewidth]{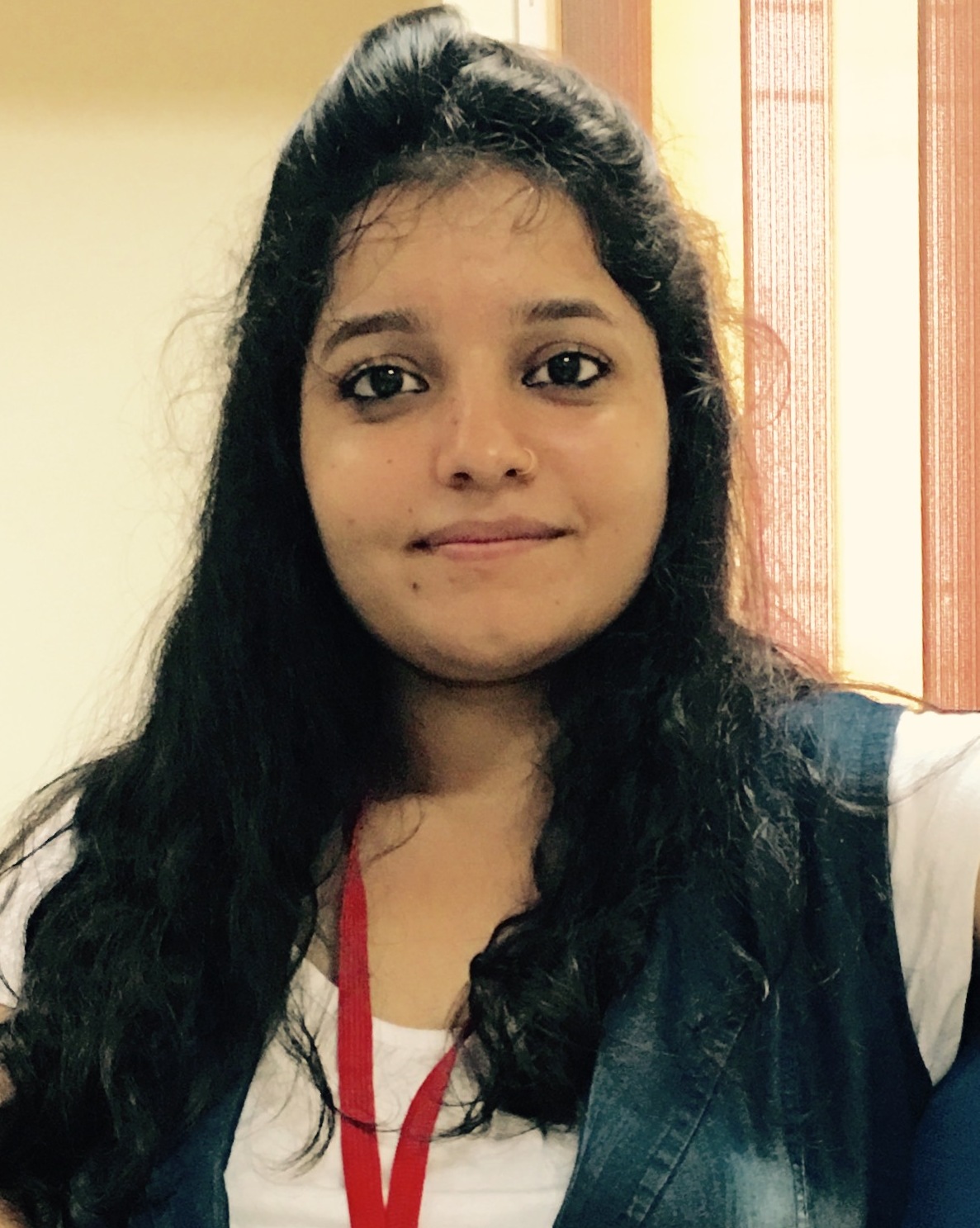}}]{Ishu Gupta} is working as a Ramanujan Faculty Fellow at Computer Science Department, IIIT, Bangalore, India. She completed her Ph.D. from NIT Kurukshetra, India with prestigious UGC-JRF-SRF Fellowships and Postdoc at Cloud Computing Research Center, NSYSU, Kaohsiung, Taiwan. Her major research interests include Cloud Computing, Cybersecurity, Artificial Intelligence, Quantum Machine Learning. She is recipient of Gold Medal for her master's degree and the ‘Excellent Paper Award’ (Twice).

\end{IEEEbiography}
\vskip 0pt plus -1fil 
\begin{IEEEbiography}[{\includegraphics[width=0.7\linewidth]{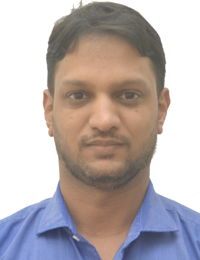}}]{Rishabh Gupta}
 received the MCA degree in Computer Science from Guru Jambheshwar University Science and Technology, Hisar, India and Ph.D. from the Department of Computer Applications, National Institute of Technology (NIT), Kurukshetra, India. Currently, he is a Post-doc fellow at The University of Aizu,  Japan. He is awarded the Senior Research Fellowship by the University Grants Commission, Government of India. His research interests include cloud computing, machine learning, and information security and privacy.\end{IEEEbiography}
\vskip 0pt plus -1fil 
\begin{IEEEbiography}[{\includegraphics[width=0.7\linewidth]{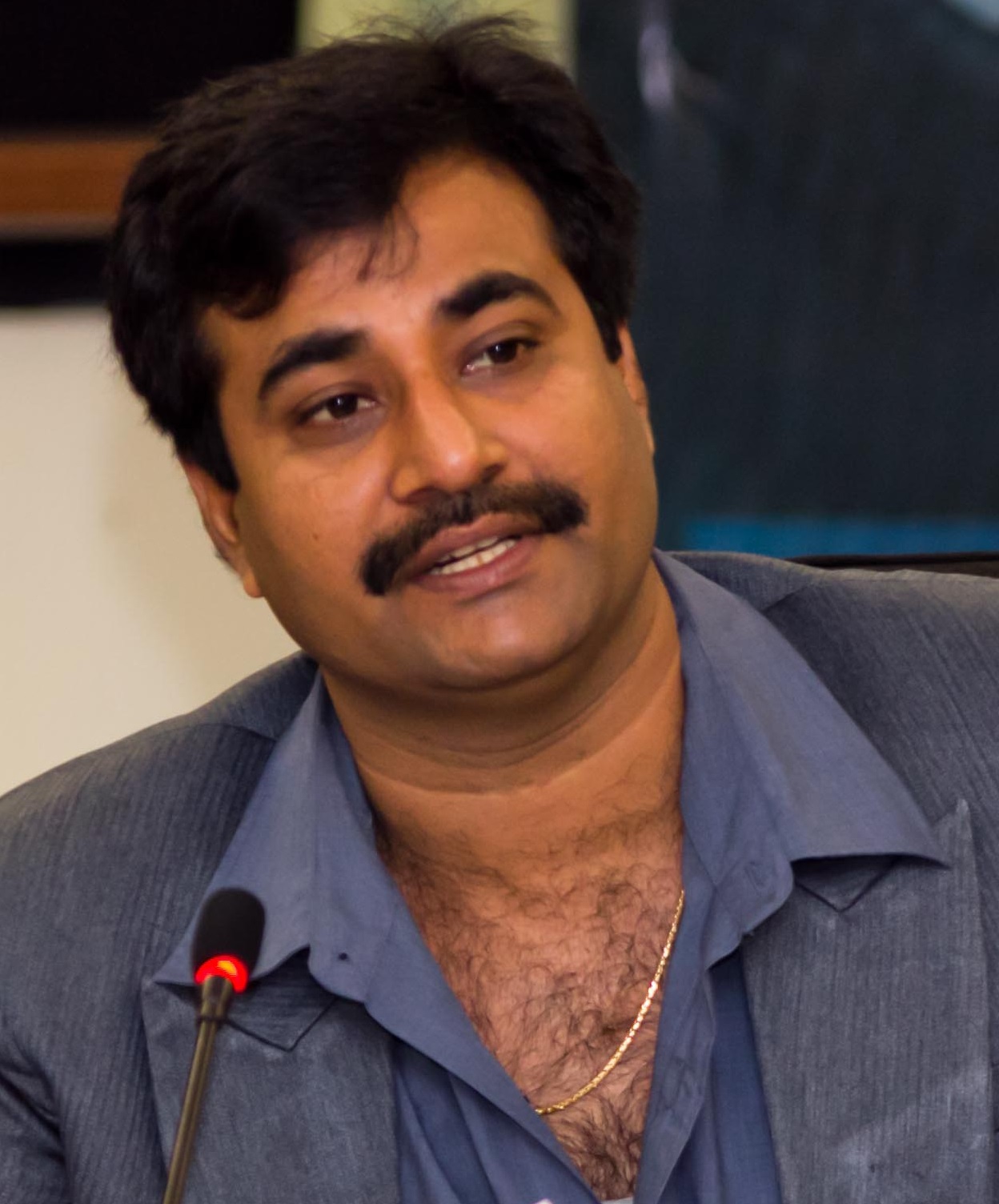}}]{Ashutosh Kumar Singh} is working as a Professor and Director of Indian Institute of Information Technology Bhopal, India. Also, he is working as Adjunct Professor in the University of Economics and Human Sciences, Warsaw, Poland.   He received his Ph.D. in Electronics Engineering from Indian Institute of Technology, BHU, India and Post Doc from Department of Computer Science, University of Bristol, UK. He has research and teaching experience in various Universities of the India, UK, and Malaysia. His research area includes Design and Testing of Digital Circuits, Data Science, Cloud Computing, Machine Learning, Security. He has published more than 370 research papers in different journals and conferences of high repute.  Some of his research findings are published in top cited journals such as IEEE TSC, IEEE TC, 	IEEE TSMC, IEEE TPDS, IEEE TII, IEEE TCC, IEEE Communications Letters, IEEE Networking Letters,  IEEE Design \& Test, IEEE Systems Journal, IEEE Wireless Communication Letters, IEEE TNSM, IET Electronics Letters,  FGCS, Neurocomputing, Information 	Sciences, Information Processing Letters, etc.
\end{IEEEbiography}
\vskip 0pt plus -1fil 
\begin{IEEEbiography}[{\includegraphics[width=0.7\linewidth]{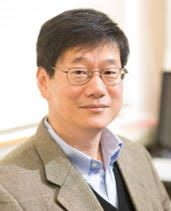}}]{Xiaoqing Wen}  received the Ph.D. degree from Osaka University, Japan, in 1993. He founded Dependable Integrated Systems Research Center in 2015 and served as its Director until 2017. His research interests include VLSI test, diagnosis, and testable design. He holds 43 U.S. Patents and 14 Japan Patents on VLSI testing. He is a fellow of the IEEE, a senior member of the IPSJ, and a senior member of the IEICE. He is serving as associate editors IEEE Transactions on VLSI and the Journal of Electronic Testing: Theory and Applications.
	
	 %He co-authored and co-edited two books: VLSI Test Principles and Architectures: Design for Testability (Morgan Kaufmann, 2006) and Power-Aware Testing and Test Strategies for Low Power Devices Springer, 2009. He holds 43 U.S. Patents and 14 Japan Patents on VLSI testing. He received the 2008 IEICE-ISS Best Paper Award for his pioneering work on X-filling-based low-capture-power test generation. He is a fellow of the IEEE, a senior member of the IPSJ, and a senior member of the IEICE. He is serving as associate editors IEEE Transactions on VLSI and the Journal of Electronic Testing: Theory and Applications.
	
\end{IEEEbiography}
% You can push biographies down or up by placing
% a \vfill before or after them. The appropriate
% use of \vfill depends on what kind of text is
% on the last page and whether or not the columns
% are being equalized.

%\vfill

% Can be used to pull up biographies so that the bottom of the last one
% is flush with the other column.
%\enlargethispage{-5in}

% that's all folks
\end{document}